\documentclass[11pt,a4paper]{article}
\usepackage{amssymb}
\usepackage{amsmath, amscd}
\usepackage{amsfonts}
\usepackage{graphicx}
\usepackage{a4wide}
\usepackage{microtype}
\usepackage{bbm}
\usepackage{slashed}
\usepackage{times}
\usepackage{amsthm}
\usepackage{mathrsfs}
\usepackage{hyperref}
\usepackage{color}

\definecolor{darkred}{rgb}{0.8,0.1,0.1}
\hypersetup{
     colorlinks=true,         
     linkcolor=darkred,
     citecolor=blue,
}

\usepackage{comment}
\usepackage[margin=2.35cm]{geometry}
\usepackage[all,cmtip]{xy}

\usepackage{marvosym}

\theoremstyle{plain}
\newtheorem{theo}{Theorem}[section]
\newtheorem{lem}[theo]{Lemma}
\newtheorem{propo}[theo]{Proposition}
\newtheorem{cor}[theo]{Corollary}

\theoremstyle{definition}
\newtheorem{defi}[theo]{Definition}

\newtheorem{rem}[theo]{Remark}

\numberwithin{equation}{section}

\def\Aff{\mathsf{Aff}}

\def\PrBu{\mathsf{PrBuGlobHyp}}
\def\VeBu{\mathsf{VeBuGlobHyp}}
\def\AfBu{\mathsf{AfBuGlobHyp}}

\def\PreSymp{\mathsf{PreSymp}}
\def\PhaseSpace{\mathfrak{PhSp}}

\def\CCR{\mathfrak{CCR}}
\def\CAR{\mathfrak{CAR}}
\def\astAlg{{^\ast}\mathsf{Alg}}
\def\QFT{\mathfrak{A}}

\def\nn{\nonumber}

\def\bbR{\mathbb{R}}
\def\bbC{\mathbb{C}}
\def\bbN{\mathbb{N}}

\def\bbZ{\mathbb{Z}}

\def\bfA{\mathsf{A}}
\def\bfV{\mathsf{V}}

\def\EE{\mathcal{E}}

\def\Hom{\mathrm{Hom}}

\def\Con{\mathrm{Con}}

\def\Gau{\mathrm{Gau}}
\def\hor{\mathrm{hor}}
\def\eqv{\mathrm{eqv}}

\def\id{\mathrm{id}}
\def\supp{\mathrm{supp}}

\def\dd{\mathrm{d}}
\def\vol{\mathrm{vol}}

\def\dim{\mathrm{dim}}

\def\1{\mathbbm{1}}
\def\oone{\mathbf{1}}
\def\g{\mathfrak{g}}

\def\ad{\mathrm{ad}}

\def\o{\mathfrak{o}}
\def\t{\mathfrak{t}}

\def\MW{\mathsf{MW}}

\newcommand{\ip}[2]{\left\langle #1,#2 \right\rangle}

\newcommand{\sect}[2]{\Gamma^\infty( #2 )}
\newcommand{\sectn}[2]{\Gamma_0^\infty(  #2 )}

\def\sk{\vspace{2mm}}

\title{%
Quantized Abelian principal connections on Lorentzian manifolds
}

\author{%
Marco Benini$^{1,3,a}$, Claudio Dappiaggi$^{1,b}$ and Alexander Schenkel$^{2,c}$\vspace{4mm}\\
{\small $^1$ Dipartimento di Fisica}\\ 
{\small Universit{\`a} di Pavia \& INFN, sezione di Pavia –- Via Bassi 6, 27100 Pavia, Italy.}\vspace{2mm}\\
{\small $^2$ Fachgruppe Mathematik}\\
{\small Bergische~Universit\"at~Wuppertal,~Gau\ss stra\ss e~20,~42119~Wuppertal,~Germany.}\vspace{4mm}
\\
{\small $^3$ II. Institut f\"ur Theoretische Physik}\\
{\small Universit\"at Hamburg,~Luruper Chaussee 149,~22761~Hamburg,~Germany.}\vspace{4mm}
\\
 {\small  ~$^a$ marco.benini@pv.infn.it~,~$^b$ claudio.dappiaggi@unipv.it~,~$^c$ schenkel@math.uni-wuppertal.de }
 }

\date{\today}

\begin{document}

\maketitle

\begin{abstract}
We construct a covariant functor from a category of Abelian principal bundles over globally hyperbolic spacetimes to a category
of $\ast$-algebras that describes quantized principal connections. 
We work within an appropriate differential geometric setting by
 using the bundle of connections and we study the full gauge group, namely the group of vertical principal bundle automorphisms. 
Properties of our functor are investigated in detail and, similar to earlier works, it is found that due 
to topological obstructions the locality property of 
locally covariant quantum field theory is violated. 
Furthermore, we prove that, for Abelian structure groups containing a nontrivial compact factor,
the gauge invariant Borchers-Uhlmann algebra of the vector dual of the bundle of connections
is not separating on gauge equivalence classes of principal connections. 
We introduce a topological generalization of the concept of locally covariant quantum fields.
As examples, we construct for the category of principal $U(1)$-bundles
 two natural transformations from singular homology functors to the quantum field theory functor
that can be interpreted as the Chern class and the electric charge.
In this case we also prove that the electric charges can be consistently set to zero, which yields
another quantum field theory functor that satisfies all axioms of locally covariant quantum field theory.
\end{abstract}


\section{\label{sec:intro}Introduction}
The algebraic theory of quantum fields on Lorentzian manifolds has made tremendous developments
since the introduction of the principle of general local covariance by Brunetti, 
Fredenhagen and Verch \cite{Brunetti:2001dx}, see also \cite{Fewster:2011pe}.
 Mathematically, this principle states that any
reasonable quantum field theory has to be formulated by a covariant functor from a category of globally hyperbolic
Lorentzian manifolds (spacetimes) to a category of unital $(C)^\ast$-algebras, subject to certain physical conditions.
Many examples of linear quantum field theories satisfying the axioms of locally covariant quantum field theory 
have been constructed in the literature, see e.g.~\cite{Bar:2007zz,Bar:2011iu,Benini:2013} and references therein. 
The mathematical tool used in these constructions is the theory of Green-hyperbolic operators
on vector bundles over spacetimes together with the $\CCR$ and $\CAR$ quantization functors.
In our previous work \cite{Benini:2012vi}  we have generalized these constructions to classes of operators on affine bundles
over spacetimes. 
In addition to these exactly tractable models, the techniques of locally covariant quantum field theory
are essential for the perturbative construction of interacting quantum field theories, see for example \cite{Brunetti:2009qc},
and the generalization of the spin-statistics theorem from Minkowski spacetime to general spacetimes \cite{Verch:2001bv}.
\sk

One of the weak points of the current status of algebraic quantum field theory is our incomplete understanding 
of the formulation of gauge theories.
Even though there exist important results on the quantization of electromagnetism
 \cite{Dimock,Pfenning:2009nx,DL,Dappiaggi:2011cj,SDH12},
linearized general relativity \cite{Fewster:2012bj} and generic linear gauge theories \cite{Hack:2012dm},
as well as on the perturbative quantization of interacting gauge theories 
\cite{Hollands:2007zg, Fredenhagen:2011mq}, there are still open problems that deserve a detailed study.
In particular, there is up to now no satisfactory formulation of quantized electromagnetism
for the following two reasons: Firstly, applying canonical quantization techniques it has been found that
electromagnetism violates the locality axiom of locally covariant quantum field theory. 
This has been shown for the field strength algebra in \cite{DL} and for the vector potential algebra  in \cite{SDH12}.
The latter reference also gives an interpretation of this feature in terms of Gauss' law. Already in 
the earlier investigations on Maxwell's equations on flat spacetimes \cite{strocchi:1967fk,strocchi:1970kx,Bongaarts:1976tt}, 
the existence of non local features has been recognized as the source of major issues in the quantization procedure, 
mostly yielding obstructions to the construction of positive algebraic states. Secondly, the differential
 geometric developments over the past decades indicate that the
natural language for formulating gauge theories of Yang-Mills type is that of principal connections on principal $G$-bundles,
which includes electromagnetism by choosing $G=U(1)$.
Taking into account the principal bundle structure has far reaching consequences for the very principle of general local covariance:
Since principal connections can not be associated to spacetimes, but only to principal bundles over spacetimes,
the category of spacetimes in \cite{Brunetti:2001dx} should be replaced by a category of principal bundles over spacetimes.
This notion of general local covariance for gauge theories of Yang-Mills type
appeared recently in the discussion of the locally covariant charged Dirac field \cite{Zahn:2012dz},
where however the principal connections were assumed to be non-dynamical background fields.
Besides this new notion of general local covariance in gauge theories of Yang-Mills type,
the classical configuration space is different to the one used in previous works:
The set of principal connections does not carry a vector space structure, but it is an affine space over the vector space 
of gauge potentials. The vector space structure employed in the works
 \cite{Dimock,Pfenning:2009nx,Dappiaggi:2011cj,SDH12} comes from
a (necessarily non-unique) fixing of some reference connection, which is unnatural in differential geometry
and leads to the unnecessary question of independence of the theory on this choice \cite{Hollands:2007zg}.
\sk

We outline the structure of our paper: In Section \ref{sec:prelim} we fix the notations and review some aspects of
the theory of principal bundles and principal connections. This material is essentially
well-known in the differential geometry literature, but we require some details that go beyond standard textbook
presentations and hence are worth for being discussed. In particular, we need a full-fledged study of the bundle
of connections \cite{Atiyah} together with the action of principal bundle morphisms and the gauge group 
(the group of vertical principal bundle automorphisms) defined on it. Sections of the bundle of connections, that is an affine bundle
over the base space, are in bijective correspondence with principal connection forms on the total space, 
but they have the advantage of being fields on the base space and not on the total space. This has far reaching consequences
when one studies dynamical equations of connections and causality properties, since the total space 
is not equipped with a Lorentzian metric.

In Section \ref{sec:phasespace} we associate to any Abelian principal bundle a gauge invariant phase space
for its principal connections by extending ideas from \cite{Benini:2012vi} and \cite{Hack:2012dm}. 
Our notion of gauge invariance is dictated by the principal bundle and in the general case differs
from the one employed in \cite{Dimock,Pfenning:2009nx,Dappiaggi:2011cj,SDH12}. The phase space is not symplectic, but
only a presymplectic  vector space, whose radical contains topological information to be discussed in Section
\ref{sec:nattrafo}. 

We characterize explicitly the gauge invariant phase space and its radical in Section
\ref{sec:radical} by using techniques from cohomology. This leads to two interesting observations: Firstly, the gauge invariant
 phase space and its radical for theories with a compact Abelian structure group exhibit a different structure with respect to their counterparts with a non-compact Abelian structure group.
Secondly, if the Abelian structure group contains a compact factor, then the gauge invariant phase 
space is not separating on gauge equivalence classes of principal connections. In particular, gauge inequivalent
flat connections can not be resolved. The reason for  this feature is that our gauge invariant phase 
space consists of affine functionals, but for Abelian structure groups with compact factors the set of  gauge 
equivalence classes of principal connections is in general no longer an affine space.
This shows that in these cases the standard phase space of affine functionals 
introduced in \cite{Benini:2012vi} has to be extended in order to be separating.
Natural candidates for this extension are Wilson loops, which are however too singular for a straightforward
description in algebraic quantum field theory. We will come back to this issue in future investigations.

The results above are combined in Section \ref{sec:phasespacefunctor} to construct
a covariant functor from a category of Abelian principal bundles over spacetimes to a category of presymplectic vector spaces.
Composing this functor with the usual $\CCR$-functor we obtain a quantum field theory functor
that satisfies the causality property and the time-slice axiom. However, the locality property of
\cite{Brunetti:2001dx} is violated, confirming that the results of \cite{DL,SDH12} also hold true in
our principal bundle geometric approach. This result was not obvious from the beginning, since our concept of
morphisms and configuration space is different from the ones in earlier investigations.

In Section \ref{sec:nattrafo} we extend the concept of a locally covariant quantum field developed
in \cite{Brunetti:2001dx} to what we call a `generally covariant topological quantum field'.
By this we mean a natural transformation from a functor describing topological information
to the quantum field theory functor. For the category of principal $U(1)$-bundles
we provide two explicit examples where the functor
describing topological information is a singular homology functor. 
The natural transformations are then the coherent association of observables
that measure the Chern class of the principal bundle and the electric charge, that is a certain cohomology class.

Following the electric charge interpretation of the previous section (see also \cite{SDH12} for an earlier account)
we show in Section \ref{sec:chargezerofunctor} that the electric charges can be consistently set to zero.
This is physically motivated since in pure electromagnetism, without the presence of charged fields,
there cannot be electric charges.
The resulting quantum field theory functor then satisfies in addition to the causality property and the time-slice axiom
also the locality property. With this we succeed in constructing a locally covariant quantum field theory.


\section{\label{sec:prelim}Geometric preliminaries}
In this work all manifolds are $C^\infty$, Hausdorff and second-countable.
Unless stated otherwise, all maps between manifolds are $C^\infty$.
Furthermore, we assume that all manifolds are of finite-type, i.e.~they possess a finite good cover.
This is a sufficient, however not necessary, condition for finite dimensional cohomology groups
and the validity of Poincar{\'e} duality, see e.g.~\cite[Chapter I, \S 5]{Bott}.
Poincar{\'e} duality will be frequently used in our work.


\subsection{Spacetimes}
We briefly review some standard notions of Lorentzian geometry, see \cite{Bar:2007zz,Bar:2011iu,Waldmann}
for a more detailed discussion.
\sk

A {\bf Lorentzian manifold} is a triple $(M,\o,g)$, where $M$ is a manifold, $\o$ is an orientation on $M$
 and $g$ is a Lorentzian metric on $M$ of signature $(-,+,\dots,+)$. 
The orientation is necessary to construct a Hodge operator.
Given also a time-orientation $\t$ on a Lorentzian manifold $(M,\o,g)$,
we call the quadruple $(M,\o,g,\t)$ a {\bf spacetime}.
Let $(M,\o,g,\t)$ be a spacetime and $S\subseteq M$ be a subset.
We denote the {\bf causal future/past} of $S$ in $M$ by
$J_M^\pm(S)$. Furthermore, $J_M(S) := J_M^+(S) \cup J_M^-(S)$.
 The subset $S\subseteq M$ is called {\bf causally compatible}, if 
$J_S^\pm(\{x\}) = J_M^\pm(\{x\})\cap S$, for all $x\in S$.
A {\bf Cauchy surface} in a spacetime $(M,\o,g,\t)$ is a subset $\Sigma\subseteq M$, which is met
exactly once by every inextensible timelike curve. A spacetime $(M,\o,g,\t)$ is called {\bf globally hyperbolic},
if it contains a Cauchy surface.


\subsection{\label{subsec:prbu}Principal bundles}
We briefly  review standard notions of principal bundles and refer to the textbooks
\cite{Kobayashi,Baum:2009zz} for more details.
\begin{defi}\label{prinGbun}
Let $M$ be a manifold and $G$ a  Lie group.
A {\bf principal $G$-bundle over $M$} is a pair $(P,r)$, where $P$ is a manifold
and $r:P\times G\to P\,,~(p,g)\mapsto r_g(p) =: p\,g$ is a smooth right $G$-action, such that
\begin{itemize}
\item[(i)] the right $G$-action $r$ is free,
\item[(ii)] $M=P/G$ is the quotient of the $G$-action $r$
and the canonical projection $\pi:P\to M$ is smooth,
\item[(iii)] $P$ is locally trivial, that is, there exists for every $x\in M$ an open neighborhood $U\subseteq M$
and a diffeomorphism $\psi: \pi^{-1}[U] \to U\times G$, which is $G$-equivariant, i.e.,~for all
$p\in\pi^{-1}[U]$ and $g\in G$, $\psi(p\,g) = \psi(p)\,g$, and fibre preserving, i.e.~$\mathrm{pr}_1\circ\psi = \pi$.
The right $G$-action on $U\times G$ is the following: For all $x\in U$ and $g,g^\prime\in G$, $(x,g)\,g^\prime := (x,g\,g^\prime)$
and $\mathrm{pr}_1:U\times G\to U$ denotes the canonical projection on the first factor.
\end{itemize}
We call $P$ the {\bf total space}, $M$ the {\bf base space}, $G$ the {\bf structure group} and
$\pi$ the {\bf projection}. 
\end{defi}

\begin{defi}
Let $G$ be a Lie group. 
For $i=1,2$, let $M_i$ be a manifold and
$(P_i,r_i)$ a principal $G$-bundle over $M_i$.
A {\bf principal $G$-bundle map} is a $G$-equivariant map $f: P_1\to P_2$,
i.e., for all $p\in P_1$ and $g\in G$, $f(p\,g) = f(p)\,g$.
\end{defi}
\begin{rem}
Notice that for any principal $G$-bundle map $f:P_1\to P_2$ there exists a unique smooth map
$\underline{f}:M_1\to M_2$, such that the following diagram commutes:
\begin{flalign}
\xymatrix{
P_1\ar[d]_-{\pi_1}\ar[rr]^-{f} && P_2\ar[d]^-{\pi_2}\\
M_1 \ar[rr]^-{\underline{f}} && M_2
}
\end{flalign}
\end{rem}

We can now define a suitable category of principal bundles over globally hyperbolic spacetimes.
\begin{defi}\label{def:prbucat}
Let $G$ be a Lie group. 
The category $G{-}\PrBu$ consists of the following objects and morphisms:
\begin{itemize}
\item An object is a tuple $\Xi=\big((M,\o,g,\t),(P,r)\big)$, where $(M,\o,g,\t)$
is a globally hyperbolic spacetime and $(P,r)$ is a principal $G$-bundle over $M$.
\item A morphism $f:\Xi_1\to \Xi_2$ is a principal $G$-bundle
 map $f:P_1\to P_2$, such that $\underline{f}:M_1\to M_2$ is an orientation and 
 time-orientation preserving isometric embedding
with $\underline{f}[M_1]\subseteq M_2$ causally compatible and open.
\end{itemize}
\end{defi}
\sk

Given any smooth left $G$-action $\rho: G\times N\to N\,,~(g,\xi)\mapsto g\,\xi$ on a manifold $N$
we can construct a covariant functor from $G{-}\PrBu$ to the category of fibre bundles over globally hyperbolic
spacetimes. This is the well-known associated bundle construction. If $N$ is further
a vector space and $\rho$ a linear representation we obtain
a covariant functor $\rho: G{-}\PrBu\to \VeBu$, where the latter category is defined as follows:
\begin{defi}\label{def:vebucat}
The category $\VeBu$ consists of the following objects and morphisms:
\begin{itemize}
\item An object  is a pair $\mathcal{V} =\big((M,\o,g,\t), (\bfV,M,\pi_\bfV,V)\big)$, where $(M,\o,g,\t)$
is a globally hyperbolic spacetime and $ (\bfV,M,\pi_\bfV,V)$ is a vector bundle over $M$.
\item A morphism $\mathcal{V}_1 \to\mathcal{V}_2$
 is a vector bundle map $\big(f:\bfV_1\to \bfV_2,\underline{f}:M_1\to M_2\big)$,
 such that  $f\vert_x: \bfV_1\vert_x\to \bfV_2\vert_{\underline{f}(x)}$ is a vector space isomorphism, for all $x\in M_1$, and
  $\underline{f}:M_1\to M_2$ is an orientation and time-orientation preserving isometric embedding
with $\underline{f}[M_1]\subseteq M_2$ causally compatible and open.
\end{itemize}
\end{defi}

Of particular relevance for us is the {\bf adjoint bundle}. Explicitly, it is the following covariant functor 
$\ad :G{-}\PrBu \to \VeBu$: To any object $\Xi=\big((M,\o,g,\t),(P,r)\big)$ we associate
$\ad(\Xi) = \big((M,\o,g,\t),(P\times_\ad \g,M,\pi_\ad,\g)\big)$, where
$\g$ is the Lie algebra of $G$, $P\times_\ad \g := (P\times \g)/G$
is the quotient by the right $G$-action $P\times \g\times G \to P\times \g\,,~(p,\xi,g)\mapsto (p\,g,\ad_{g^{-1}}(\xi))$
and $\pi_\ad$ denotes the map obtained from the projection $P\times\g\to P$ via the quotient.
To any morphism $f:\Xi_1\to \Xi_2$ we associate $\ad(f): \ad(\Xi_1) \to \ad(\Xi_2)$,
which is the vector bundle map (covering $\underline{f}:M_1\to M_2$) given by
\begin{flalign}
\ad(f): P_1\times_\ad \g \to P_2\times_\ad \g~,~~[p,\xi]\mapsto [f(p),\xi]~.
\end{flalign}
We review the following well-known 
\begin{lem}\label{lem:trivadbund}
Let $M$ be a manifold, $G$ an Abelian Lie group and $(P,r)$ a  principal $G$-bundle over $M$.
Then $P\times_\ad \g = M\times \g$, i.e.~the adjoint bundle is trivial.
\end{lem}
\begin{proof}
Since $G$ is Abelian the adjoint action is trivial, i.e.
 $P\times_\ad \g = (P\times \g)/G = P/G\times \g = M\times \g$.
\end{proof}


\subsection{Principal connections}
Connections on principal bundles constitute the fundamental degrees of freedom
in gauge theories of Yang-Mills type. In this subsection we will review the relevant definitions and properties
following \cite{Kobayashi,Baum:2009zz}.

\begin{defi}\label{def:princon}
Let $M$ be a manifold, $G$ a Lie group and $(P,r)$ a principal $G$-bundle over $M$. 
A {\bf connection form} on $(P,r)$ is a $\g$-valued one-form $\omega \in \Omega^1(P,\g)$ satisfying
the following two conditions:
\begin{itemize}
\item[(i)] $\omega(X^\xi_p) = \xi$, for all $\xi\in\g$ and $p\in P$, where 
$X_p^\xi \in T_pP$ is the fundamental vector at $p$ corresponding to $\xi$.
\item[(ii)] $r_g^\ast(\omega) = \ad_{g^{-1}}(\omega) $, for all $g\in G$. 
\end{itemize}
We denote the set of all connection forms by $\Con(P)$.
\end{defi}

\begin{rem}
Due to \cite[Chapter II, Theorem 2.1]{Kobayashi} there exists a connection form, i.e.~$\Con(P)\neq \emptyset$.
\end{rem}

\begin{defi}
Let $\Omega^k(P,\g)$ be the vector space of $\g$-valued $k$-forms, $k= 0,\dots,\dim(P)$.
\begin{itemize}
\item[a)] We call $\eta\in\Omega^k(P,\g)$ {\bf $G$-equivariant}, if $r_g^\ast(\eta) = \ad_{g^{-1}}(\eta) $,
for all $g\in G$.
\item[b)] We call $\eta\in \Omega^k(P,\g)$ {\bf horizontal}, if $\eta(Y_1,\dots,Y_k)=0$ whenever
at least one $Y_i\in T_pP$ is vertical, i.e.~$\pi_\ast(Y_i)=0$.
\end{itemize}
The vector space of $G$-equivariant and horizontal $\g$-valued $k$-forms 
is denoted by $\Omega^k_{\hor}(P,\g)^{\eqv}$.
\end{defi}

According to \cite[Chapter II, Section 5]{Kobayashi}, see also \cite[Satz 3.5]{Baum:2009zz},
we have the following
\begin{propo}\label{propo:eqvhoriso}
Let $M$ be a manifold, $G$ a Lie group and $(P,r)$ a principal $G$-bundle over $M$. 
Then, for all $k=0,\dots,\dim(M)$, the vector space $\Omega^k_{\hor}(P,\g)^{\eqv}$ is isomorphic
 to the vector space of $P\times_\ad \g$-valued $k$-forms on $M$, $\Omega^k(M,P\times_\ad \g)$.
\end{propo}

\begin{rem}\label{rem:eqvhoriso}
Let $G$ be an Abelian Lie group. Due to Lemma \ref{lem:trivadbund} we have $P\times_\ad\g = M\times \g$
and hence $\Omega^k(M,P\times_\ad \g) = \Omega^k(M,\g)$, for all $k=0,\dots,\dim(M)$.
In this case the isomorphism of Proposition \ref{propo:eqvhoriso}  is given by the
pull-back map $\pi^\ast:\Omega^k(M,\g) \to \Omega^k_\hor(P,\g)^\eqv$.
We shall denote in the following the inverse of this map simply by an underline,
i.e.~for all $\eta\in \Omega_\hor^k(P,\g)^\eqv$, $\underline{\eta} := \pi^{\ast\,-1}(\eta)\in\Omega^k(M,\g)$.
\end{rem}

There is a canonical action of the Abelian group $\Omega^1_\hor(P,\g)^\eqv $
on $\Con(P)$,
\begin{flalign}
\Con(P)\times \Omega^1_\hor(P,\g)^\eqv \to \Con(P)~,~~(\omega,\eta)\mapsto \omega + \eta~.
\end{flalign}
This action is free and transitive, thus $\Con(P)$ is an affine space over $\Omega^1_\hor(P,\g)^\eqv $ and, due
 to Proposition \ref{propo:eqvhoriso}, also over $\Omega^1(M,P\times_\ad\g)$. For any 
 Abelian Lie group $G$, $\Con(P)$ is an affine space over $\Omega^1(M,\g)$.
\sk

In category theoretical terms, the above construction implies that there exists a contravariant functor
$\Con: G{-}\PrBu \to \Aff$, where $\Aff$ is the category of (not necessarily finite dimensional) affine spaces.
To any object $\Xi$  the functor associates
the affine space $\Con(P)$ modeled on $\Omega^1(M,P\times_\ad \g)$. To
a morphism $f:\Xi_1\to \Xi_2$ the functor associates the affine map given by restricting the
 pull-back $f^\ast: \Omega^1(P_2,\g) \to \Omega^1(P_1,\g)$.

\begin{defi}
Let $M$ be a manifold, $G$ a  Lie group and $(P,r)$ a principal $G$-bundle over $M$.
The {\bf curvature} is the following map
\begin{flalign}
\mathcal{F}: \Con(P) \to \Omega^2_\hor(P,\g)^\eqv~,~~\omega \mapsto \mathcal{F}(\omega) = 
\dd\omega +\frac{1}{2} [\omega,\omega]_\g ~,
\end{flalign}
where $\dd$ is the exterior differential  and $[\cdot,\cdot]_\g$ denotes the Lie bracket on $\g$.
\end{defi}

\begin{rem}
Let $G$ be an Abelian Lie group. Since in this case the Lie bracket $[\cdot,\cdot]_\g$ is trivial, the curvature
reads $\mathcal{F}(\omega) = \dd \omega$, for all $\omega\in \Con(P)$.
Furthermore, applying  Remark \ref{rem:eqvhoriso} we can consider equivalently the curvature as a map
\begin{flalign}\label{eqn:curvaturedown}
\underline{\mathcal{F}} : \Con(P) \to \Omega^2(M,\g)~,~~\omega \mapsto \underline{\mathcal{F}}(\omega) =
\underline{\mathcal{F}(\omega)}= \underline{\dd\omega}~.
\end{flalign}
As a consequence of the (Abelian) Bianchi identity $\dd\mathcal{ F}(\omega)= \dd \dd\omega=0$, for all $\omega\in\Con(P)$,
we obtain that $\underline{\mathcal{F}}(\omega)\in \Omega_\dd^2(M,\g)$ is closed, for all $\omega\in\Con(P)$.
\end{rem}

The next statement is valid only for Abelian Lie groups. It implies that Abelian Yang-Mills theories are not self-interacting and hence it simplifies drastically our construction of the associated quantum field theory.
\begin{lem}\label{lem:curvaffineop}
Let $M$ be a manifold, $G$ an Abelian Lie group and $(P,r)$ a principal $G$-bundle over $M$.
The map $\underline{\mathcal{F}}:\Con(P)\to\Omega^2(M,\g)$ is an affine map
with linear part $\underline{\mathcal{F}}_V: \Omega^1(M,\g) \to \Omega^2(M,\g)\,,~\eta\mapsto \dd\eta$.
\end{lem}
\begin{proof}
Let $\omega\in \Con(P)$ and $\eta\in\Omega^1(M,\g)$, then
$\underline{\mathcal{F}}(\omega + \pi^\ast(\eta)) = \underline{\dd\omega + \dd\pi^\ast(\eta)}
= \underline{\mathcal{F}}(\omega) + \underline{\pi^\ast(\dd \eta)} 
= \underline{\mathcal{F}}(\omega) +\dd\eta$.
\end{proof}
Rephrasing this statement in the language of category theory,
we obtain the following important insight: For an Abelian Lie group $G$, the curvature is a natural transformation
$\underline{\mathcal{F}}: \Con \Rightarrow \Omega^2_{\mathrm{base}}$,
where $\Omega^2_{\mathrm{base}}:G{-}\PrBu \to \Aff$ is the contravariant functor
associating to any object $\Xi$ the $\g$-valued $2$-forms on the base space 
$\Omega^2(M,\g)$ (regarded as an affine space)
and to a morphism $f:\Xi_1\to \Xi_2$ the pull-back $\underline{f}^\ast: \Omega^2(M_2,\g)\to \Omega^2(M_1,\g)$.


\subsection{The Atiyah sequence}\label{subsec:Atiyah}
We briefly review the Atiyah sequence \cite{Atiyah} using a category theoretical language.
Consider the following covariant functors from $G{-}\PrBu$ to $\VeBu$:
\begin{enumerate}
\item adjoint bundle functor $\ad$, given at the end of Section \ref{subsec:prbu}.
\item base space tangent bundle functor $T_\mathrm{base}$, with
$T_\mathrm{base}(\Xi) := \big((M,\o,g,\t),(TM,M,\pi_{TM},\bbR^{\dim(M)})\big)$ and
$T_\mathrm{base}(f):=\big(\underline{f}_\ast:TM_1\to TM_2,\underline{f}:M_1\to M_2\big)$.
\item quotient of the total space tangent bundle functor $T_{\mathrm{total}/G}$, with $T_{\mathrm{total}/G}(\Xi) := 
\big((M,\o,g,\t),(TP/G,\linebreak M,\pi\circ \pi_{TP},\bbR^{\dim(P)})\big)$
and $T_{\mathrm{total}/G}(f):=\big(f_\ast:TP_1/G\to TP_2/G,\underline{f}:M_1\to M_2\big)$.
\item trivial associated bundle functor $\rho_{0}$, with $\rho_{0}(\Xi) := \big((M,\o,g,\t),(M\times \{0\},M,\mathrm{pr}_1,\{0\})\big)$
and $\rho_0(f):= \big(\underline{f}\times\id_{\{0\}} :M_1\times\{0\} \to M_2\times\{0\},\underline{f}:M_1\to M_2\big)$.
\end{enumerate}
For the following construction let us compose the four covariant functors above with the forgetful
 functor which forgets the fibre-wise invertibility of the morphisms in the category $\VeBu$.
For keeping the notation as simple as possible we do not introduce a new symbol for the latter category
and just remember this convention for the rest of this subsection.
Then there exists a sequence of natural transformations
\begin{flalign}\label{eqn:Atiyahnattrans}
\xymatrix{
\rho_0 \ar@{=>}[r]& \ad \ar@{=>}[r] &  {T_{\mathrm{total}/G}} \ar@{=>}[r] & T_{\mathrm{base}} \ar@{=>}[r] & \rho_0~. 
}
\end{flalign}
Explicitly, for every object $\Xi = \big((M,\o,g,\t),(P,r)\big)$ in $G{-}\PrBu$, there exists
a sequence of vector bundle maps (covering $\id_M$), called the {\bf Atiyah sequence},
\begin{flalign}\label{eqn:Atiyah}
\xymatrix{
M\times \{0\} \ar[r]^-{\alpha} & P\times_\ad \g \ar[r]^-{\iota} & TP/G\ar[r]^-{\pi_\ast} & TM \ar[r]^-{\beta} & M\times \{0\}~,
}
\end{flalign}
where $\alpha(x,0) = [p,0]$, with $p\in \pi^{-1}[\{x\}]$ arbitrary, $\iota([p,\xi]) = [X^\xi_p]$, 
$\pi_\ast([Y]) = \pi_\ast(Y)$ and $\beta(X) = (\pi_{TM}(X),0)$.
The following statement is proven in \cite{Atiyah}.
\begin{propo}\label{propo:atiyahexact}
The Atiyah sequence (\ref{eqn:Atiyah}) is a short exact sequence.
\end{propo}


\subsection{The bundle of connections}
A vector bundle map $\lambda:TM\to TP/G$ covering the identity $\id_M:M\to M$ 
is called a {\bf splitting of the Atiyah sequence} (\ref{eqn:Atiyah}), if $\pi_\ast\circ\lambda = \id_{TM}$.
The set of splittings of the Atiyah sequence can be modeled by a subbundle
of the homomorphism bundle $\Hom(TM,TP/G)$, called the {\bf bundle of connections}. 
For later convenience we use once more a category theoretical language
to describe this subbundle.
\sk

As a first step, we address the construction of covariant functors describing homomorphism bundles.
Let $\mathfrak{F},\mathfrak{G} : G{-}\PrBu \to \VeBu$ be two covariant functors. Using the fact that
all morphisms in $\VeBu$ are fibre-wise invertible we can construct a covariant
functor $\Hom_{\mathfrak{F},\mathfrak{G}}:  G{-}\PrBu \to \VeBu$ as follows:
For any object $\Xi$  we set 
\begin{flalign}
\Hom_{\mathfrak{F},\mathfrak{G}}(\Xi):=\Big((M,\o,g,\t),\big(\Hom(\mathfrak{F}(\Xi),
\mathfrak{G}(\Xi)),M,\pi_{\mathfrak{F}(\Xi),\mathfrak{G}(\Xi)},\bbR^{\mathrm{rank}(\mathfrak{F}(\Xi))\times\mathrm{rank}(\mathfrak{G}(\Xi))}\big)\Big)~,
\end{flalign}
where, with a slight abuse of notation, we denoted the total space of 
the vector bundle contained in the object $\mathfrak{F}(\Xi)$ also as $\mathfrak{F}(\Xi)$.
To any morphism $f:\Xi_1\to \Xi_2$ we associate the following vector bundle map (covering
$\underline{f}:M_1\to M_2$)
\begin{flalign}
\Hom_{\mathfrak{F},\mathfrak{G}}(f) :  \Hom(\mathfrak{F}(\Xi_1),\mathfrak{G}(\Xi_1))\to 
\Hom(\mathfrak{F}(\Xi_2),\mathfrak{G}(\Xi_2))~,~~L\mapsto \mathfrak{G}(f) \circ L\circ \mathfrak{F}(f)^{-1}~. 
\end{flalign}

We are going to interpret the splitting condition in terms of a suitable natural transformation. As in Section
\ref{subsec:Atiyah}, in the construction of the natural transformation, we are dropping the condition according to which the morphisms in $\VeBu$ are fibre-wise invertible.
With the natural transformation $T_{\mathrm{total}/G}\Rightarrow T_{\mathrm{base}}$ introduced in (\ref{eqn:Atiyahnattrans})
we construct a natural transformation $\Hom_{T_\mathrm{base},T_{\mathrm{total}/G}} \Rightarrow 
\Hom_{T_\mathrm{base},T_\mathrm{base}}$
by setting, for any object $\Xi$ in $G{-}\PrBu$,
\begin{flalign}\label{eqn:nautraltemp}
l_{\pi_\ast} : \Hom(TM,TP/G) \to \Hom(TM,TM)~,~~\lambda \mapsto \pi_\ast\circ \lambda~.
\end{flalign}
We induce on the submanifold $\mathcal{C}(\Xi) := l_{\pi_\ast}^{-1}(\id_{TM})$ the structure
of a subbundle of $\Hom(TM,TP/G)$. 
We denote this subbundle by $(\mathcal{C}(\Xi),M,\pi_{\mathcal{C}(\Xi)},A^{\dim(M)\times \dim(\g)})$,
where $A^{\dim(M)\times \dim(\g)}$ is the unique (up to isomorphism) affine space 
modeled on $\bbR^{\dim(M)\times \dim(\g)}$.
As a consequence of Proposition \ref{propo:atiyahexact}, 
$(\mathcal{C}(\Xi),M,\pi_{\mathcal{C}(\Xi)},A^{\dim(M)\times \dim(\g)})$
is an affine bundle modeled on the homomorphism bundle $\Hom(TM,P\times_\ad \g)$.
Our definition of affine bundles is the one of \cite[Chapter 6.22]{KMS} and \cite[Definition 2.11]{Benini:2012vi}.
Furthermore, since (\ref{eqn:nautraltemp}) is a natural transformation,
the bundle of connections can be seen as a covariant functor $\mathcal{C}: G{-}\PrBu \to \AfBu$, 
where the latter category is defined as follows:
\begin{defi}
The category $\AfBu$ consists of the following objects and morphisms:
\begin{itemize}
\item An object is a triple $\mathcal{A}=\big((M,\o,g,\t), (\bfA,M,\pi_\bfA,A),(\bfV,M,\pi_\bfV,V)\big)$, where $(M,\o,g,\t)$
is a globally hyperbolic spacetime and $ (\bfA,M,\pi_\bfA,A)$ is an affine bundle over $M$ modeled on
the vector bundle $(\bfV,M,\pi_\bfV,V)$.
\item A morphism  $\mathcal{A}_1\to \mathcal{A}_2 $
 is a fibre bundle map $\big(f:\bfA_1\to \bfA_2,\underline{f}:M_1\to M_2\big)$,
 such that  $f\vert_x:\bfA_1\vert_x\to \bfA_2\vert_{\underline{f}(x)}$ is an affine space isomorphism, for all $x\in M_1$, and
  $\underline{f}:M_1\to M_2$ is an orientation and time-orientation preserving isometric embedding
with $\underline{f}[M_1]\subseteq M_2$ causally compatible and open.
\end{itemize}
\end{defi}
\begin{rem}
Every morphism $(f,\underline{f})$ in $\AfBu$ determines a unique vector bundle map between the underlying vector bundles
(that is a morphism in $\VeBu$) by taking fibre-wise the linear part. We call this vector bundle map with a slight abuse of notation
 the linear part of $(f,\underline{f})$ and denote it by $(f_V,\underline{f})$. 
 \end{rem}

\subsection{Sections of the bundle of connections}
The set of sections $\sect{M}{\mathcal{C}(\Xi)}$ of the bundle of connections
is an affine space modeled on the vector space $\sect{M}{\Hom(TM,P\times_\ad \g)}$, cf.~\cite[Lemma 2.20]{Benini:2012vi}. 
The latter is isomorphic to the $P\times_\ad \g$-valued one-forms on $M$, i.e.~$\Omega^1(M,P\times_\ad\g)$.
We follow the usual abuse of notation and denote by $\lambda + \eta$ the
action of $\eta\in \Omega^1(M,P\times_\ad\g)$ on $\lambda\in \sect{M}{\mathcal{C}(\Xi)}$.
In category theoretical terms, the above construction is a contravariant functor
$\Gamma^\infty \circ \mathcal{C} : G{-}\PrBu \to \Aff$. To any object
$\Xi$ the functor associates the affine space $\sect{M}{\mathcal{C}(\Xi)}$ modeled on 
$\Omega^1(M,P\times_\ad\g)$. To a morphism $f:\Xi_1\to \Xi_2$ the functor associates the affine map
\begin{flalign}\label{eqn:connectionmorph}
\Gamma^\infty(\mathcal{C}(f)) : \sect{M_2}{\mathcal{C}(\Xi_2)} \to \sect{M_1}{\mathcal{C}(\Xi_1)}~,~~
\lambda\mapsto \mathcal{C}(f)^{-1} \circ \lambda \circ \underline{f}~.
\end{flalign}
This is exactly the pull-back of a section $\lambda\in \sect{M_2}{\mathcal{C}(\Xi_2)}  $ to $\sect{M_1}{\mathcal{C}(\Xi_1)} $
via the affine bundle map $\mathcal{C}(f)$.
With a slight abuse of notation we shall denote this pull-back also simply by
$f^\ast(\lambda) := \Gamma^\infty(\mathcal{C}(f)) (\lambda)$.
\sk

We can define for any connection form $\omega\in \Con(P)$ an element $\lambda_\omega\in\sect{M}{\mathcal{C}(\Xi)}$
by, for all $X\in TM$, $\lambda_\omega(X) := [X^{\uparrow_\omega}_p]$. 
The arrow symbol denotes the horizontal lift  with respect to $\omega$ of $X\in TM$ to an arbitrary $p\in \pi^{-1}[\{\pi_{TM}(X)\}]$.
For each object $\Xi$ in $G{-}\PrBu$ this construction provides us with a map 
$\Con(P) \to \sect{M}{\mathcal{C}(\Xi)}\,,~\omega\mapsto \lambda_\omega$.
Using the explicit expressions, the following statement descends directly:
\begin{propo}\label{prop:affineiso}
The maps defined above yield a natural isomorphism $\Con \Rightarrow \Gamma^\infty\circ \mathcal{C}$.
\end{propo}

Let now $G$ be an Abelian Lie group. Due to Lemma \ref{lem:curvaffineop} (and the text below this lemma)
 the curvature can be regarded as a natural transformation
$\underline{\mathcal{F}}: \Con \Rightarrow \Omega^2_\mathrm{base}$.
Using the natural isomorphism of Proposition \ref{prop:affineiso}
we obtain a natural transformation (denoted with a slight abuse of notation also 
by the symbol  $\underline{\mathcal{F}}$) $\underline{\mathcal{F}}: \Gamma^\infty\circ \mathcal{C} \Rightarrow  
\Omega^2_\mathrm{base}$.
Explicitly, we obtain for any object $\Xi$ in $G{-}\PrBu$ 
an affine map 
$\underline{\mathcal{F}}: \sect{M}{\mathcal{C}(\Xi)} \to  \Omega^2(M,\g)$
with linear part $\underline{\mathcal{F}}_V : \Omega^1(M,\g)\to \Omega^2(M,\g)\,,~\eta \mapsto -\dd\eta$. 
(The minus sign is part of the natural isomorphism of Proposition \ref{prop:affineiso}.)
According to \cite[Section 3]{Benini:2012vi} this is an affine differential operator.
\sk

We conclude this section by studying gauge transformations.
\begin{defi}
Let $M$ be a manifold, $G$ a  Lie group and $(P,r)$ a principal $G$-bundle over $M$.
A {\bf gauge transformation}  is a $G$-equivariant diffeomorphism 
$f:P\to P$, such that $\underline{f}=\id_M$.
We denote by $\Gau(P)$ the group of all gauge transformations of $(P,r)$.
\end{defi}
Notice that whenever $\Xi=\big((M,\o,g,\t), (P,r)\big)$ is an object in  $G{-}\PrBu$, 
a gauge transformation $f\in\Gau(P)$  is an automorphism  in the same category.
\sk

Let now $G$ be an Abelian Lie group. Then the gauge group $\Gau(P)$
is isomorphic to the group $C^\infty(M,G)$.
This isomorphism is constructed as follows:
Notice that for any $f\in \Gau(P)$ there exists a unique $\widetilde{f}\in C^\infty(P,G)$,
such that, for all $p\in P$, $f(p) = p\,\widetilde{f}^{-1}(p)$ (the use of the inverse is purely conventional).
Since $f$ is $G$-equivariant and $G$ is Abelian the map $\widetilde{f}$ has to be
$G$-invariant and hence it defines a unique element $\widehat{f}\in C^\infty(M,G)$.
A straightforward calculation shows that map $f\mapsto \widehat{f}$ is an isomorphism of groups.
The action of the gauge group on sections of the bundle of connections can be derived
from the general discussion of morphisms that we have given above and we obtain
\begin{flalign}
 \sect{M}{\mathcal{C}(\Xi)} \times C^\infty(M,G) \to\sect{M}{\mathcal{C}(\Xi)}~,~~(\lambda,\widehat{f})\mapsto
\lambda + \widehat{f}^\ast(\mu_G)~,
\end{flalign}
where $\mu_G\in \Omega^1(G,\g)$ is the Maurer-Cartan form. Notice that the gauge group
acts in terms of affine maps on $\sect{M}{\mathcal{C}(\Xi)}$ and that
the linear part of these maps is the identity.


\section{\label{sec:phasespace}The phase space for an object}
Let $G$ be an Abelian Lie group and $\Xi=\big((M,\o,g,\t),(P,r)\big)$ an object in $G{-}\PrBu$.
Let $\mathcal{C}(\Xi)$ be the associated bundle of connections
and $\sect{M}{\mathcal{C}(\Xi)}$ the affine space of sections.
We denote the vector dual bundle (see \cite[Definition 2.15]{Benini:2012vi})
 by $\mathcal{C}(\Xi)^\dagger$ and by $\sectn{M}{\mathcal{C}(\Xi)^\dagger}$
  the vector space of compactly supported sections. 
 The aim of this section is to construct a gauge invariant phase space for dynamical 
 principal connections on $\Xi$.
\sk

The dynamics is governed by Maxwell's equations, which 
are described in our setting by the affine differential operator
\begin{flalign}
\MW := \delta \circ \underline{\mathcal{F}} : \sect{M}{\mathcal{C}(\Xi)} \to \Omega^1(M,\g)~,~~\lambda \mapsto \MW(\lambda) =\delta \underline{\mathcal{F}}(\lambda)~,
\end{flalign}
where $\delta$ is the codifferential and $\underline{\mathcal{F}}$ is the 
curvature affine differential operator.
The linear part of  $\MW$ is 
\begin{flalign}
\MW_V : \Omega^1(M,\g) \to \Omega^1(M,\g)~,~~\eta \mapsto \MW_V(\eta) = \delta\underline{\mathcal{F}}_V(\eta) = -\delta\dd\eta~.
\end{flalign}
Due to \cite[Theorem 3.5]{Benini:2012vi}, the affine differential operator $\MW$ is formally adjoinable
to a differential operator $\MW^\ast : \Omega_0^1(M,\g^\ast)\to \sectn{M}{\mathcal{C}(\Xi)^\dagger}$, with 
$\g^\ast$ denoting the vector space dual of the Lie algebra $\g$. 
Explicitly, $\MW^\ast$ is determined (up to the ambiguities to be discussed below) by the condition,
for all $\lambda\in\sect{M}{\mathcal{C}(\Xi)} $ and $\eta  \in \Omega_0^1(M,\g^\ast)$,
\begin{flalign}\label{eqn:adjointaffop}
\ip{\eta}{\MW(\lambda)}:=\int_M \eta \wedge \ast\big(\MW(\lambda)\big)= \int_M \vol\,\big(\MW^\ast(\eta)\big)(\lambda)~,
\end{flalign}
where $\ast$ denotes the Hodge operator and $\vol$ the volume form.
We will always suppress the duality pairing between $\g^\ast$ and $\g$ in order to simplify the notation.
\sk

As it is proven in \cite[Theorem 3.5]{Benini:2012vi}, the formal adjoint differential operator
$\MW^\ast: \Omega_0^1(M,\g^\ast)\to \sectn{M}{\mathcal{C}(\Xi)^\dagger}$ is not unique. Uniqueness is
restored if we quotient out the trivial elements\footnote{By trivial we mean that the corresponding classical affine observables
(\ref{eqn:affinefunctional}), i.e.~functionals on the configuration space $\sect{M}{\mathcal{C}(\Xi)}$, vanish.}
\begin{flalign}
\mathrm{Triv} := \Big\{a\,\1 \in \sectn{M}{\mathcal{C}(\Xi)^\dagger} : a\in C^\infty_0(M) \text{ satisfies }\int_M\vol \,a =0\Big\}~,
\end{flalign}
i.e.~if we consider the operator $\MW^\ast:\Omega_0^1(M,\g^\ast)\to \sectn{M}{\mathcal{C}(\Xi)^\dagger}/\mathrm{Triv} $.
By $\1 \in \sect{M}{\mathcal{C}(\Xi)^\dagger}$ we denote the canonical section which associates to every $x\in M$
the constant affine map in the fibre $\mathcal{C}(\Xi)^\dagger\vert_x$ 
defined as $\1(\lambda)=1$ for each $\lambda\in\mathcal{C}(\Xi)\vert_x$.
In the following we shall use the convenient notation $\EE^\mathrm{kin}:= \sectn{M}{\mathcal{C}(\Xi)^\dagger}/\mathrm{Triv} $.
The quotient by $\mathrm{Triv}$ does not affect the linear part of $\MW^\ast(\eta)$:
Indeed,  for all $\eta\in\Omega^1_0(M,\g^\ast)$, $\lambda\in\sect{M}{\mathcal{C}(\Xi)}$ and $\eta^\prime\in\Omega^1(M,\g)$,
\begin{flalign}
\nn \int_M \vol\,\big(\MW^\ast(\eta)\big)\big(\lambda + \eta^\prime\big) &= 
\ip{\eta}{\MW\big(\lambda + \eta^\prime\big)} = \ip{\eta}{\MW(\lambda) - \delta\dd\eta^\prime}\\
&= \int_M \vol\,\big(\MW^\ast(\eta)\big)(\lambda) + \ip{-\delta\dd\eta}{\eta^\prime}
\end{flalign}
implies that the linear part 
is $\MW^\ast(\eta)_{V} = -\delta\dd\eta$, for all $\eta\in\Omega^1_0(M,\g^\ast)$.
\sk

The next step is to restrict to those
elements in $\EE^\mathrm{kin}$ that describe gauge invariant observables.
It is enlightening to introduce the vector space of classical affine observables 
$\{\mathcal{O}_\varphi: \varphi\in\EE^\mathrm{kin}\}$, where 
$\mathcal{O}_\varphi$ is the functional on the configuration space $\sect{M}{\mathcal{C}(\Xi)} $ defined by
\begin{flalign}\label{eqn:affinefunctional}
\mathcal{O}_\varphi: \sect{M}{\mathcal{C}(\Xi)} \to \bbR~,~~\lambda \mapsto \mathcal{O}_\varphi(\lambda) = \int_M\vol~\varphi\big(\lambda\big)~.
\end{flalign}
Let $\widehat{f}\in C^\infty(M,G)\simeq \Gau(P)$ be an element in the gauge group.
The gauge transformations on $\sect{M}{\mathcal{C}(\Xi)}$
 are given by $\lambda \mapsto\lambda + \widehat{f}^\ast(\mu_G)$.
Demanding invariance of $\mathcal{O}_\varphi$ under gauge transformations,
 i.e.~$\mathcal{O}_{\varphi}\big(\lambda + \widehat{f}^\ast(\mu_G)\big) 
=\mathcal{O}_\varphi(\lambda) $ 
 for all $\lambda\in \sect{M}{\mathcal{C}(\Xi)}$ and $\widehat{f}\in C^\infty(M,G)$, leads to the
  following condition for the linear part $\varphi_{V}\in \Omega^1_0(M,\g^\ast)$
 of $\varphi \in \EE^\mathrm{kin}$, for all $\widehat{f}\in C^\infty(M,G)$,
 \begin{flalign}
 \ip{\varphi_{V}}{\widehat{f}^\ast(\mu_G)} =0~.
 \end{flalign}
 This provides the motivation for the following vector subspace
 \begin{flalign}\label{eqn:gaugeinv}
 \EE^{\mathrm{inv}} := \Big\{\varphi\in\EE^\mathrm{kin} 
 : \ip{\varphi_{V}}{\widehat{f}^\ast(\mu_G)} = 0~,~~\forall \widehat{f}\in C^\infty(M,G) \Big\}\subseteq \EE^\mathrm{kin}~,
 \end{flalign}
 which serves as a starting point to construct the phase space.
\begin{lem}\label{lem:coclosed}
\begin{itemize}
\item[a)] For all $\varphi \in \EE^{\mathrm{inv}}$ the linear part $\varphi_{V}\in \Omega^1_0(M,\g^\ast)$ is coclosed, 
i.e.~$\delta\varphi_{V} =0$. 
\item[b)] All $\varphi\in\EE^{\mathrm{kin}}$ satisfying
$\varphi_V = \delta \eta$ for some $\eta\in \Omega^2_0(M,\g^\ast)$ are elements in $\EE^\mathrm{inv}$.
\end{itemize}
\end{lem}
\begin{proof}
 Proof of a): 
Let  $\chi\in C^\infty(M,\g)$ and consider the element of the gauge group specified
by $\widehat{f}_\chi := \exp \circ \chi \in C^\infty(M,G)$, where $\exp:\g\to G$ denotes the exponential map.
The pull-back of the Maurer-Cartan form then reads $\widehat{f}_\chi^\ast(\mu_G) = \dd\chi$.
Let $\varphi\in \EE^{\mathrm{inv}}$ be arbitrary. Due to (\ref{eqn:gaugeinv}) the linear part $\varphi_{V}$ of $\varphi$
satisfies, for all $\chi\in C^\infty(M,\g)$,
\begin{flalign}
0=\ip{\varphi_{V}}{\widehat{f}^\ast_\chi(\mu_G)} = \ip{\varphi_{V}}{\dd\chi} = \ip{\delta\varphi_{V}}{\chi}~,
\end{flalign}
which implies $\delta\varphi_{V}=0$.
\sk

Proof of b): For all $\widehat{f}\in C^\infty(M,G)$, 
\begin{flalign}
\ip{\varphi_V}{\widehat{f}^\ast(\mu_G)} = \ip{\delta\eta}{\widehat{f}^\ast(\mu_G)}
= \ip{\eta}{\dd \widehat{f}^{\ast}(\mu_G)}= \ip{\eta}{\widehat{f}^{\ast}(\dd \mu_G)}= 0~,
\end{flalign}
 since the Maurer-Cartan form of Abelian Lie groups is closed.
\end{proof}
\begin{cor}\label{cor:gaugeinvobs}
Let us define the vector spaces
\begin{subequations}
\begin{flalign}
\EE^\mathrm{min} &:= \big\{  \varphi \in \EE^\mathrm{kin} : \varphi_V \in \delta \Omega^2_0(M,\g^\ast)\big\}~~,\\
\EE^\mathrm{max} &:= \big\{  \varphi \in \EE^\mathrm{kin} : \varphi_V \in \Omega^1_{0,\delta}(M,\g^\ast)\big\}~~.
\end{flalign}
\end{subequations}
Then the following inclusions of vector spaces hold true
\begin{flalign}
\EE^\mathrm{min}\subseteq \EE^\mathrm{inv} \subseteq \EE^\mathrm{max}~.
\end{flalign}
\end{cor}
\begin{rem}
This corollary provides us with a lower and upper bound on the vector space $\EE^\mathrm{inv}$. Notice that in case
 $M$ has a trivial first de Rham cohomology group $H_{\mathrm{dR}}^1(M,\g)=\{0\}$
  (which implies that the dual cohomology group is trivial
  $H^1_{0\,\mathrm{dR}^\ast} (M,\g^\ast):=\Omega^1_{0,\delta}(M,\g^\ast)/\delta \Omega^2_0(M,\g^\ast) =\{0\}$),
the lower and upper bounds coincide, i.e.~$\EE^\mathrm{min}=\EE^\mathrm{inv} = \EE^\mathrm{max}$.
The explicit characterization of $\EE^\mathrm{inv}$  will be postponed to Section 
\ref{sec:radical}.
\end{rem}

\sk

The equation of motion $\MW(\lambda)=0$ is implemented at a dual level on  $\EE^\mathrm{inv}$
by considering the quotient vector space $\EE:=\EE^\mathrm{inv}/\MW^\ast\big[\Omega_0^1(M,\g^\ast)\big]$.
To construct a presymplectic structure on this space let us consider the Hodge-d'Alembert operators
$\square_{(k)} :=  \delta\circ \dd + \dd\circ\delta  :\Omega^k(M,\g^\ast)\to \Omega^k(M,\g^\ast)$, that are normally hyperbolic
operators. The corresponding unique retarded/advanced Green's operators are denoted by
$G^{\pm}_{(k)}:\Omega_0^k(M,\g^\ast)\to \Omega^k(M,\g^\ast) $ and the causal propagators are defined
by $G_{(k)}:= G_{(k)}^+ - G_{(k)}^-:\Omega_0^k(M,\g^\ast)\to \Omega^k(M,\g^\ast) $. We notice the relations
\begin{subequations}
\begin{flalign}
\square_{(k)}\circ \dd = \dd \circ \square_{(k-1)} ~,\quad \square_{(k)}\circ \delta = \delta \circ \square_{(k+1)}~,
\end{flalign}
which, together with formal self-adjointness of $\square_{(k)}$, imply
\begin{flalign}\label{useful}
G_{(k)}^\pm \circ \dd = \dd\circ G_{(k-1)}^\pm~,\quad G_{(k)}^\pm\circ \delta = \delta\circ G_{(k+1)}^{\pm}~.
\end{flalign}
\end{subequations}
Let us further assume that we are given a bi-invariant pseudo-Riemannian metric $h$ 
on the Lie group $G$. This structure is equivalent to an $\ad$-invariant
 inner product (possibly indefinite) on the Lie algebra
$\g$ and hence a vector space isomorphism (denoted with a slight abuse of notation by the same symbol) 
$h:\g\to \g^\ast$. Notice that a metric $h$ is necessary to specify
a Lagrangian and hence a Poisson bracket, cf.~Remark \ref{rem:peierls}.
We denote by $h^{-1}:\g^\ast \to \g$ the inverse vector space isomorphism.
 Using also the pairing $\ip{~}{~}$ we define for all $\eta,\eta^\prime\in \Omega^k(M,\g^\ast)$ with 
compact overlapping support the non-degenerate (indefinite) inner product
\begin{flalign}
\ip{\eta}{\eta^\prime}_h := \ip{\eta}{h^{-1}(\eta^\prime)}~.
\end{flalign}
As a consequence of $\square_{(k)}$ being formally self-adjoint, 
$G_{(k)}$ turns out to be formally skew-adjoint with respect to $\ip{~}{~}_h$.

\begin{propo}\label{propo:presymp}
Let $G$ be an Abelian Lie group and $h$ a bi-invariant pseudo-Riemannian
metric on $G$. Let further $\Xi=\big((M,\o,g,\t),(P,r)\big)$ be an object in $G{-}\PrBu$. Then the vector space 
$\EE := \EE^\mathrm{inv}/\MW^\ast\big[\Omega_0^1(M,\g^\ast)\big]$ can be equipped with the presymplectic structure
\begin{flalign}\label{eqn:presymp}
\tau: \EE\times \EE \to \bbR~,~~([\varphi],[\psi])\mapsto \tau([\varphi],[\psi]) = \ip{\varphi_{V}}{G_{(1)}(\psi_{V})}_h~.
\end{flalign}
In other words, $(\EE,\tau)$ is a presymplectic vector space.
\end{propo}
\begin{proof}
We have to prove that $\tau$ is well-defined, i.e.~that for every $\varphi= \MW^\ast(\eta)$, $\eta\in\Omega^1_0(M,\g^\ast)$,
 we have $\ip{\varphi_{V}}{G_{(1)}(\psi_{V})}_h =0 $ and $\ip{\psi_{V}}{G_{(1)}(\varphi_{V})}_h=0$ for the linear parts $\psi_{V}$ 
of all elements $\psi\in \EE^\mathrm{inv}$.
Lemma \ref{lem:coclosed} implies that $\delta\psi_{V}=0$. The first property holds true:
\begin{flalign}
\nn\ip{\varphi_{V}}{G_{(1)}(\psi_{V})}_h &= \ip{\MW^\ast(\eta)_{V}}{G_{(1)}(\psi_{V})}_h = 
-\ip{\delta\dd\eta}{G_{(1)}(\psi_{V})}_h\\
\nn&=-\ip{\eta}{\delta\dd G_{(1)}(\psi_{V})}_h = -\ip{\eta}{(\square_{(1)} - \dd\delta)\big( G_{(1)}(\psi_{V})\big)}_h\\
&=\ip{\eta}{\dd G_{(0)}(\delta\psi_{V}) }_h = 0~.
\end{flalign}
The second property follows analogously, since $G_{(1)}$ 
is formally skew-adjoint  with respect to $\ip{~}{~}_h$.
 From the latter property it also follows that $\tau$ is antisymmetric.
\end{proof}
\begin{rem}\label{rem:peierls}
The presymplectic structure (\ref{eqn:presymp}) can be derived from a Lagrangian form
by generalizing the method of Peierls \cite{Peierls:1952cb} 
to gauge theories. This generalization has already been studied in \cite{Marolf:1993af}  and it was 
put on mathematically solid grounds recently in \cite{SDH12} for the vector potential of $U(1)$-connections.
Since in our approach the configuration space $\sect{M}{\mathcal{C}(\Xi)}$ is different,
we have to adapt the relevant arguments to our setting: 
Let us consider the Lagrangian form $\mathcal{L}[\lambda] := 
-\frac{1}{2} h\big(\underline{\mathcal{F}}(\lambda)\big) \wedge \ast\big( \underline{\mathcal{F}}(\lambda)\big)$ 
and its perturbation  by an element $\varphi \in \EE^\mathrm{inv}$, 
i.e.~$\mathcal{L}_\varphi [\lambda] := \mathcal{L}[\lambda] + \vol\,\varphi(\lambda)$. 
Notice that a bi-invariant metric $h$ on $G$ is required in order to define the Lagrangian.
The Euler-Lagrange equation corresponding to $\mathcal{L}_\varphi$ is 
$\MW(\lambda) + h^{-1}(\varphi_{V}) =0$, where $\varphi_{V}\in \Omega_0^1(M,\g^\ast)$ is the linear part of $\varphi$.
Let us take any  $\lambda\in \sect{M}{\mathcal{C}(\Xi)}$ satisfying $\MW(\lambda)=0$.
 The goal is to construct the retarded/advanced effect of $\varphi$ on this solution.
Let $\Sigma^\pm\subset M$ be two Cauchy surfaces (with $\Sigma^+$ being in the future of $\Sigma^-$) 
such that $\supp(\varphi_{V}) \subseteq J^-_M\big(\Sigma^+\big) \cap J^+_M\big(\Sigma^-\big)$
 (this means that $\varphi_{V}$ has support in the spacetime region between $\Sigma^+$ and $\Sigma^-$).
We are looking for a $\lambda_\varphi^\pm \in \sect{M}{\mathcal{C}(\Xi)}$ satisfying the equation of motion
$\MW(\lambda_\varphi^\pm) + h^{-1}(\varphi_{V}) =0$ and $\lambda_\varphi^\pm \vert_{J_M^\mp(\Sigma^\mp)}  = 
 (\lambda + \widehat{f}_\pm^\ast(\mu_G))\vert_{J_M^\mp(\Sigma^\mp)} $
  for some $\widehat{f}_\pm\in C^\infty(M,G)$. The latter condition states that $\lambda_\varphi^\pm$
agrees up to a gauge transformation with $\lambda$ in the past/future of $\Sigma^\mp$.
Since $\sect{M}{\mathcal{C}(\Xi)}$ is an affine space over $\Omega^1(M,\g)$ we find a unique
$\eta_\varphi^\pm\in \Omega^1(M,\g)$ such that $\lambda_\varphi^\pm = \lambda + \eta_\varphi^\pm$.
The equations of motion for $\lambda$ and $\lambda_\varphi^\pm$ then
 imply $-\delta\dd\eta_\varphi^\pm + h^{-1}(\varphi_{V}) =0$
and the asymptotic condition reads $\big(\eta_\varphi^\pm -\widehat{f}_\pm^\ast(\mu_G)\big)\big 
\vert_{J_M^\mp(\Sigma^\mp)} =0$ for some $\widehat{f}_\pm\in C^\infty(M,G)$.
By gauge equivalence,
we can assume without loss of generality
that $\eta_\varphi^\pm$ satisfies $\delta\eta_\varphi^\pm =0$, and hence the equation of motion reads
$\square_{(1)}(\eta_\varphi^\pm) = h^{-1}(\varphi_{V})$. For the support condition
 $\eta_\varphi^\pm\vert_{J_M^\mp(\Sigma^\mp)} =0$
(that is contained in the asymptotic condition above) the unique solution of this equation is 
$\eta_\varphi^\pm = G_{(1)}^\pm\big(h^{-1}(\varphi_V)\big) = h^{-1}\big(G^\pm_{(1)}(\varphi_V)\big)$.
All solutions of the equation $-\delta\dd\eta_\varphi^\pm + h^{-1}( \varphi_{V}) =0$
subject to the asymptotic condition  $\big(\eta_\varphi^\pm -\widehat{f}_\pm^\ast(\mu_G)\big)\big 
\vert_{J_M^\mp(\Sigma^\mp)} =0$, for some $\widehat{f}_\pm\in C^\infty(M,G)$, are obtained by adding 
a pure gauge solution $\widehat{f}^\ast(\mu_G)$ to
 $\eta_\varphi^\pm = h^{-1}\big(G_{(1)}^\pm(\varphi_{V})\big)$. Let now $\psi \in \EE^\mathrm{inv}$ and consider the
gauge invariant functional $\mathcal{O}_\psi$ as in (\ref{eqn:affinefunctional}). The retarded/advanced effect
of $\varphi\in \EE^\mathrm{inv}$ on $\mathcal{O}_\psi$ is defined by $E^\pm_\varphi\big(\mathcal{O}_\psi\big)(\lambda):=
 \mathcal{O}_\psi(\lambda_\varphi^\pm) - \mathcal{O}_\psi(\lambda) 
 = \ip{\psi_{V}}{\eta_\varphi^\pm} = \ip{\psi_{V}}{h^{-1}\big(G_{(1)}^\pm(\varphi_{V})\big)}
 =\ip{\psi_V}{G_{(1)}^\pm(\varphi_V)}_h$. 
 Notice that this expression is well-defined
since $\mathcal{O}_\psi$ is gauge invariant. We find that the presymplectic structure (\ref{eqn:presymp}) is given by the
 difference of the retarded and advanced effect, i.e.~$\tau([\psi],[\varphi]) = E^+_\varphi\big(\mathcal{O}_\psi\big)(\lambda)  - 
 E^-_\varphi\big(\mathcal{O}_\psi\big)(\lambda)$, which agrees with the idea of Peierls \cite{Peierls:1952cb}.
\end{rem}
\sk

We come to the characterization of the radical $\mathcal{N}\subseteq \EE$ 
of the presymplectic structure $\tau$. An element $[\psi]\in \EE$ is in $\mathcal{N}$ if and only if,
for all $[\varphi]\in \EE$,
$\tau([\varphi],[\psi]) =0$. In this section we will  only provide
a lower and upper estimate for the vector space $\mathcal{N}$. The explicit characterization 
will be content of Section \ref{sec:radical}.
\begin{lem}\label{radical}
\begin{itemize}
\item[a)] 
Let  $[\psi]\in \mathcal{N}$ be arbitrary. Then any representative $\psi\in \EE^\mathrm{inv}$ is such that $\psi_V = \delta \alpha$
for some $\alpha \in \Omega_{0,\dd}^2(M,\g^\ast)$.
\item[b)]Let $\psi\in \EE^\mathrm{inv}$ be such that $\psi_V = \delta\dd \gamma$
with $\gamma\in\Omega_{\mathrm{tc}}^{1}(M,\g^\ast)$ and $\dd \gamma\in \Omega_0^{2}(M,\g^\ast)$.
Then $[\psi]\in\mathcal{N}$. The subscript $_{\mathrm{tc}}$ denotes forms of timelike compact support.
\end{itemize}
\end{lem}
\begin{proof}
Proof of a): By hypothesis $[\psi]$ satisfies, for all $[\varphi]\in \EE$, 
\begin{flalign}\label{eqn:tmpnullspace}
\tau([\varphi],[\psi]) = \ip{\varphi_V}{G_{(1)}(\psi_V)}_h=0~.
\end{flalign}
By Corollary \ref{cor:gaugeinvobs} we have that $\EE^\mathrm{min}\subseteq \EE^\mathrm{inv}$
and thus it is necessary for $[\psi]$ to fulfill, for all $\eta\in\Omega_0^2(M,\g^\ast)$,
\begin{flalign}
0=\ip{\delta \eta}{ G_{(1)}(\psi_V)}_h = \ip{\eta}{G_{(2)}(\dd\psi_V)}_h~.
\end{flalign}
This implies that $G_{(2)}(\dd\psi_V) =0$ and hence due to the fact that $G_{(2)}$ is the causal propagator of
a normally hyperbolic operator we obtain $\dd\psi_V = \square_{(2)}(\alpha)$ for some $\alpha\in \Omega_0^2(M,\g^\ast)$.
Applying $\dd$ to this equation shows that $\dd\alpha =0$, i.e.~$\alpha\in \Omega_{0,\dd}^2(M,\g^\ast)$.
Applying $\delta$ and using that $\delta\psi_V=0$ (cf.~Lemma \ref{lem:coclosed}) we find
$\square_{(1)}(\psi_V) = \square_{(1)}(\delta\alpha)$. This implies $\psi_V = \delta\alpha$ and completes the proof.\sk

Proof of b): Let $\psi\in \EE^\mathrm{inv}$ be as above. For all $\varphi\in \EE^\mathrm{inv}$, we obtain
\begin{flalign}
\nn \tau([\varphi],[\psi]) & = \ip{\varphi_V}{G_{(1)}(\delta\dd\gamma)}_h=
 \ip{\varphi_V}{\delta\dd G_{(1)}(\gamma)}_h\\
\nn &= \ip{\varphi_V}{(\square_{(1)} -\dd\delta)\big(G_{(1)}(\gamma)\big)}_h =  -\ip{\varphi_V}{\dd\delta G_{(1)}(\gamma)}_h\\
 &=  -\ip{\delta\varphi_V}{\delta G_{(1)}(\gamma)}_h = 0~.
\end{flalign}
In the second equality we exploited the possibility to enlarge the domain of $G_{(1)}$
 to $\Omega^1_{\mathrm{tc}}(M,\g^\ast)$ \cite{SDH12}
 and in the last equality we used the identity $\delta\varphi_V=0$.
\end{proof}
\begin{cor}\label{cor:radicalinclusions}
Let us define the vector spaces
\begin{subequations}
\begin{flalign}
\mathcal{N}_{\mathrm{min}} &:= \big\{\psi\in \EE^\mathrm{inv}: \psi_V \in \delta \big(\Omega_0^{2}(M,\g^\ast)\cap \dd \Omega_\mathrm{tc}^1(M,\g^\ast)\big) \big\}\big /\MW^\ast\big[\Omega^1_0(M,\g^\ast)\big]~,\\
\mathcal{N}_{\mathrm{max}}&:=\big\{\psi\in\EE^\mathrm{inv} : \psi_V \in \delta \Omega_{0,\dd}^2(M,\g^\ast)\big\}\big /\MW^\ast\big[\Omega^1_0(M,\g^\ast)\big]~.
\end{flalign}
\end{subequations}
Then the following inclusions of vector spaces hold true
\begin{flalign}
\mathcal{N}_{\mathrm{min}} \subseteq \mathcal{N} \subseteq \mathcal{N}_{\mathrm{max}}\subseteq
\EE~.
\end{flalign}
\end{cor}

\begin{rem}
The radical $\mathcal{N}$ of the theory under consideration is in general different from that of affine matter field theories,
see \cite[Proposition 4.4]{Benini:2012vi}. Even though the constant affine observables
$[a\,\1]$,  with $a\in C^\infty_0(M)$, are contained in $\mathcal{N}$, in general they do not exhaust
all elements.  The lower bound on $\mathcal{N}$, given in Corollary \ref{cor:radicalinclusions}, 
coincides with the radical obtained in \cite{SDH12} (up to the constant affine observables
 which are not present in this last mentioned paper, since it does not exploit the complete 
 geometric structure of the bundle of connections).
\end{rem}
\begin{rem}\label{rem:Nmin}
If $M$ has compact Cauchy surfaces all elements in $\mathcal{N}_\mathrm{min}$
have representatives with trivial linear part. In general this last statement does not hold true, as the following example proves:
Let us consider the case in which $G=\bbR$ (implying $\g^\ast =\bbR$) and
 $M$ is diffeomorphic to $\mathbb{R}^2\times\mathbb{S}^{m-2}$, where $m \geq 4$
 and $\mathbb{S}^{m-2}$ denotes the ${m-2}$-sphere (we suppress this diffeomorphism in the following). 
Any Cauchy surface $\Sigma\subseteq M$ is diffeomorphic to $\mathbb{R}\times\mathbb{S}^{m-2}$.
 Since $H^1_{0\,\mathrm{dR}}(\bbR) = \bbR$ is nontrivial, we can find an $\alpha\in \Omega^1_{0,\dd}(\bbR)$
which is not exact. Let us introduce Cartesian coordinates $(t,x)$ on the $\bbR^2$ factor of
$M$. We denote by $\alpha_t\in \Omega^1_{\dd}(M)$ the pull-back of $\alpha$
along the projection to the time coordinate $t$ and by $\alpha_x \in \Omega^1_{\dd}(M)$ the pull-back
of $\alpha$ along the projection to the space coordinate $x$.
We define $\eta := \alpha_t\wedge \alpha_x$.
The support property of $\alpha$ and the compatibility between $\dd$ and the pull-backs entail that
 $\eta\in\Omega^2_{0,\dd}(M)$. 
 Furthermore, since $H^1_\mathrm{dR}(M)=\{0\}$, there exists a $\beta \in C^\infty(M)$ 
 such that $\alpha_x=-\dd\beta$, which implies $\eta = \dd(\beta \,\alpha_t)$,
 where $\beta\,\alpha_t\in\Omega^1_{\mathrm{tc}}(M)$. We now show that $\eta\notin \dd\Omega^1_0(M)$:
 Let $\nu_{\mathbb{S}^{m-2}}$ be the normalized volume form on $\mathbb{S}^{m-2}$ and
  let $\mathrm{pr}:M\to\mathbb{S}^{m-2}$ be the projection from $M$ to  $\mathbb{S}^{m-2}$.
  Notice that the integral $\int_M \eta\wedge \mathrm{pr}^\ast(\nu_{\mathbb{S}^{m-2}})= \big(\int_\bbR \alpha\big)^2\neq 0$
  does not vanish,
  since $\alpha$ is not exact. If there would exist a $\gamma\in \Omega^1_0(M)$, such that $\eta =\dd\gamma$,
  then by Stokes' theorem the integral would vanish, which is a contradiction. Hence, $\eta = \dd(\beta \,\alpha_t)$,
  with $\beta\,\alpha_t\in\Omega^1_{\mathrm{tc}}(M)$, defines a nontrivial element in $H^2_{0\,\mathrm{dR}}(M)$.
  Furthermore, for the class in $\mathcal{N}_{\mathrm{min}}$ defined by 
  $\underline{\mathcal{F}}^\ast(\eta)\in \EE^\mathrm{inv}$ there exists no representative with a trivial
  linear part: Indeed,
  suppose that there exists $\gamma\in\Omega^1_0(M)$ such that $\underline{\mathcal{F}}^\ast(\eta)_V = 
  -\delta\eta  =\delta \dd\gamma$. 
  Using that $\eta$ is closed and of compact support, this equation entails $-\square_{(2)}(\eta )=\square_{(2)} ( \dd\gamma)$
  which yields the contradiction $\eta=-\dd\gamma$, since $\square_{(2)}$ is a normally hyperbolic operator.
\end{rem}


\section{\label{sec:radical}Explicit characterization of $\EE^\mathrm{inv}$ and $\mathcal{N}$}
So far we obtained only upper and lower bounds for the vector spaces $\EE^\mathrm{inv}$ and $\mathcal{N}$, see
Corollary \ref{cor:gaugeinvobs} and Corollary \ref{cor:radicalinclusions}. The goal of this section is
 to provide an explicit characterization of $\EE^\mathrm{inv}$ and $\mathcal{N}$
 when $G$ is a connected Abelian Lie group.
 Due to \cite[Theorem 2.19]{Adams} the latter assumption entails that $G$ is isomorphic
 to $\mathbb{T}^k\times \bbR^l$, for some $k,l\in\bbN_0$.
 \sk
 
For this endeavor we have to understand more explicitly
 how the gauge group $\mathrm{Gau}(P)\simeq C^\infty(M,G)$ acts on $\sect{M}{\mathcal{C}(\Xi)}$. 
Let us consider the homomorphism of Abelian groups
\begin{flalign}\label{eqn:homtemp}
C^\infty(M,G) \to \Omega^1(M,\g)~,~~\widehat{f}\mapsto \widehat{f}^\ast(\mu_G)~.
\end{flalign}
Notice that the image of this homomorphism characterizes the action of the gauge group on $\sect{M}{\mathcal{C}(\Xi)}$. 
Using that the Maurer-Cartan form of any Abelian Lie group is closed, we obtain that the homomorphism 
(\ref{eqn:homtemp}) maps to the closed one-forms $\Omega^1_\dd(M,\g)$.
Using further that $\exp[C^\infty(M,\g)]$ is an Abelian subgroup of $C^\infty(M,G) $
and that the image of $\exp[C^\infty(M,\g)]$ under the group homomorphism (\ref{eqn:homtemp}) is 
$\dd C^\infty(M,\g)$, we arrive at an injective Abelian group homomorphism
\begin{flalign}
C^\infty(M,G) / \exp[C^\infty(M,\g)] \to H^1_\mathrm{dR}(M,\g)~,~~[\widehat{f}]\mapsto [\widehat{f}^\ast(\mu_G)]~.
\end{flalign}
We denote the image of this homomorphism by $A_G\subseteq H^1_\mathrm{dR}(M,\g)$.
Notice that the Abelian group $A_G$ characterizes exactly the gauge transformations which
are not of exponential form $\exp\circ \chi$, for some $\chi\in C^\infty(M,\g)$.
  \sk
  
  Since any connected Abelian Lie group $G$ is isomorphic to $\mathbb{T}^k\times \mathbb{R}^{l}$, 
  the map $\widehat{f}\in C^\infty(M,G)$ is given by a $k+l$-tuple of maps
  $\big(\widehat{f}_1,\dots,\widehat{f}_{k+l}\big)$, where $\widehat{f}_i\in C^\infty(M,\mathbb{T})$, for $i=1,\dots,k$,
  and $\widehat{f}_i\in C^\infty(M,\mathbb{R})$, for $i=k+1,\dots,k+l$.
  The Abelian group $C^\infty(M,G) / \exp[C^\infty(M,\g)] $ factorizes into the direct  product 
  $\big(C^\infty(M,\mathbb{T})/\exp[C^\infty(M,i\,\bbR)] \big) ^{k} \times \big(C^\infty(M,\bbR)/\exp[C^\infty(M,\bbR)] \big)^l $,
  where $i\,\bbR$ is the Lie algebra of $\mathbb{T}$ and $\bbR$ is the Lie algebra of $\bbR$.
  Also the cohomology group splits into a direct sum $H^1_\mathrm{dR}(M,\g) = 
  H^1_\mathrm{dR}(M,i\bbR)^{\oplus_k} \oplus H^1_\mathrm{dR}(M,\bbR)^{\oplus_l}$. The Abelian group $A_G$ is thus given by
  a direct sum of Abelian groups $A_G = A_\mathbb{T} ^{\oplus_k}\oplus A_{\mathbb{R}}^{\oplus_l}$ 
  (remember that the direct product and direct sum of groups over a finite index set yield the same group).
  In this way the problem of characterizing $A_G$ is reduced to 
  the problem of characterizing $A_\mathbb{T}$ and $A_{\mathbb{R}}$.
\begin{propo}
$A_{\bbR} = \{0\}$.
\end{propo}  
  \begin{proof}
  This is a consequence of the exponential map $\exp :\g \to G $ being an isomorphism
  for $G=\bbR$.
  \end{proof}
  
To characterize $A_\mathbb{T}$ we are using techniques from sheaf cohomology, see 
e.g.~\cite[Section 4]{alggeo}.
Let us denote by $C^\infty_M(-,H)$ the sheaf of smooth functions on $M$ with values in an Abelian Lie group
$H$. Explicitly, for any open subset $U\subseteq M$ the sheaf associates the Abelian
group $C^\infty(U,H)$ of smooth functions on $U$ with values in $H$. 
In the following we shall require for $H$ the choices $2\pi i\,\bbZ$ (regarded as a zero-dimensional Lie group), 
$i\,\bbR$ (with group operation
given by addition) and $\mathbb{T}$ (with group operation given by multiplication).
There is an exact sequence of sheaves 
\begin{flalign}
\xymatrix{
0 \ar[r] & C^\infty_M(-,2\pi i\,\bbZ)  \ar[r]^\iota &C^\infty_M(-,i\,\bbR)\ar[r]^{\exp} & C^\infty_M(-,\mathbb{T}) \ar[r] & 0~,
}
\end{flalign}
where $\iota$ is induced from the canonical injection $2\pi i\,\bbZ \hookrightarrow i\,\bbR$
and $\exp$ is induced from the exponential map $i\,\bbR \to \mathbb{T}$.
Remember that by definition a sequence of sheaves is called exact if and only if 
for each point $x\in M$ the sequence of stalks is exact. In particular, this definition does not require
that $C^\infty(U,i\,\bbR)\to C^\infty(U,\mathbb{T}) $ is surjective, for all open subsets $U\subseteq M$, but only that this
property holds for sufficiently small $U$. The obstruction to surjectivity of this group homomorphism for $U=M$
is encoded in the long exact sequence of sheaf cohomology. 
For the present case we are interested in the following part of the aforementioned sequence:
\begin{flalign}
\cdots\, \to\, C^\infty(M,i\bbR) \to  C^\infty(M,\mathbb{T}) \to  H^1(M,C^\infty_M(-,2\pi i\,\bbZ))
\to  H^1(M,C^\infty_M(-,i\,\bbR)) \,\to \,\cdots 
\end{flalign}
Notice that the sheaf $C^\infty_M(-,i\,\bbR)$ is soft (i.e.~every real valued function defined on a closed subset of
$M$ can be extended to $M$) and as a consequence the sheaf cohomology group vanishes, 
i.e.~$H^1(M,C^\infty_M(-,i\,\bbR))  = 0$.
This implies that the group homomorphism
$C^\infty(M,\mathbb{T}) \to  H^1(M,C^\infty_M(-,2\pi i\,\bbZ))$ is
surjective, hence $A_\mathbb{T}\simeq C^\infty(M,\mathbb{T})/\exp[C^\infty(M,i\,\bbR)] \simeq 
H^1(M,C^\infty_M(-,2\pi i\,\bbZ))$.
Since $M$ is a manifold, the sheaf cohomology group $H^1(M,C^\infty_M(-,2\pi i\,\bbZ))$
is isomorphic to the first singular cohomology group 
as well as to the first {\v C}ech cohomology group with coefficients in $2\pi i\,\bbZ$.
We simply use the symbol $H^1(M,2\pi i\,\bbZ) := H^1(M,C^\infty_M(-,2\pi i\,\bbZ))$.
In summary, we have the following
\begin{propo}
$A_\mathbb{T} \simeq H^1(M,2\pi i\,\mathbb{Z})$.
\end{propo}

\begin{cor}\label{cor:AGgeneral}
$A_G \simeq H^1(M,2\pi i\,\mathbb{Z})^{\oplus_k}$.
\end{cor}

We now provide an explicit characterization of $\EE^\mathrm{inv}$.
This is based on the following relation between cohomology with real and integer coefficients:
 \begin{flalign}\label{eqn:coefficientiso}
H^1(M,2\pi i\,\bbZ)\otimes_\bbZ \bbR \simeq H^1(M,i\,\bbR)~.
 \end{flalign}
A proof of this result can be found in \cite[Chapter 7.1.1]{Voisin} under the assumption that the manifold $M$ 
is of finite type, which is our case as specified at the beginning of Section \ref{sec:prelim}.
\begin{theo}\label{theo:gaugeinvspace}
Let $G\simeq \mathbb{T}^k\times \bbR^l$ be a connected Abelian Lie group and
$\Xi=\big((M,\o,g,\t),(P,r)\big)$ an object in $G{-}\PrBu$.
Then the gauge invariant subspace $\EE^\mathrm{inv}$ (\ref{eqn:gaugeinv}) is 
\begin{flalign}
\EE^\mathrm{inv}  =\big\{ \varphi \in \EE^\mathrm{kin}:
 \varphi_V \in \delta \Omega^2_0(M,i\bbR)^{\oplus_k}\oplus \Omega^1_{0,\delta}(M,\bbR)^{\oplus_l}\big\} ~.
\end{flalign}
\end{theo}
\begin{proof}
By definition, $\EE^\mathrm{inv}$ is the  vector subspace of $\EE^\mathrm{kin}$, such that
the linear parts annihilate $\{\widehat{f}^\ast(\mu_G):\widehat{f}\in C^\infty(M,G)\}$. Due to Corollary  \ref{cor:gaugeinvobs}
we have that $\EE^\mathrm{inv}\subseteq \EE^\mathrm{max}=\{\varphi \in \EE^\mathrm{kin} :
 \varphi_V\in \Omega^1_{0,\delta}(M,\g^\ast)\}$ and hence we can pair the linear parts of
 elements $\varphi\in \EE^\mathrm{inv}$ with cohomology classes $[\eta]\in H_{\mathrm{dR}}^1(M,\g)$, 
 $\ip{\varphi_V}{[\eta]} =\int_M \varphi_V \wedge \ast(\eta) $. The gauge invariance condition amounts to
 $\ip{\varphi_V}{A_G} =\{0\}$, for all $\varphi\in\EE^\mathrm{inv}$, and by Corollary \ref{cor:AGgeneral} this is equivalent to
 \begin{flalign}\label{eqn:tempEinv}
 \ip{\varphi_V}{H^1(M,2\pi i\,\mathbb{Z})^{\oplus_k}}=\{0\}~.
 \end{flalign}
 Since  $H^1_\mathrm{dR}(M,i\,\bbR) \simeq H^1(M,i\,\bbR)
 \simeq H^1(M,2\pi i\,\mathbb{Z}) \otimes_\bbZ \bbR$
 and since the map $\ip{\varphi_V}{~} : H^1_\mathrm{dR}(M,i\,\bbR) \to \bbR$ is linear, (\ref{eqn:tempEinv})
 implies that, for all $\varphi\in \EE^\mathrm{inv}$,
 \begin{flalign}\label{noflat}
 \ip{\varphi_V}{H_{\mathrm{dR}}^1(M,i\bbR)^{\oplus_k}}=\{0\}~.
 \end{flalign}
 As a consequence of Poincar{\'e} duality, 
 $\varphi_V \in \delta \Omega^2_0(M,i\bbR)^{\oplus_k}\oplus \Omega^1_{0,\delta}(M,\bbR)^{\oplus_l}$ which completes the proof.
\end{proof}
\begin{rem}
Notice that if $G\simeq \mathbb{T}^k\times \bbR^l$ contains a nontrivial compact factor (i.e.~$k>0$),
the vector space of gauge invariant classical affine functionals $\{\mathcal{O}_\varphi: \varphi \in \EE^\mathrm{inv}\}$
(cf.~(\ref{eqn:affinefunctional})) does not separate all gauge equivalence classes of connections:
Given two connections $\lambda_1,\lambda_2\in \sect{M}{\mathcal{C}(\Xi)}$ with the same curvature, then
 there exists $\eta \in \Omega^1_\dd(M,\g)$ such that $\lambda_2 = \lambda_1 +\eta$.
 Let us assume that $[\eta] \in H_{\mathrm{dR}}^1(M,i\bbR)^{\oplus_k} \subseteq H_{\mathrm{dR}}^1(M,\g)$, but
 $[\eta]\not\in A_G$ such that $\lambda_1$ and $\lambda_2$ are not gauge equivalent
  (this exists e.g.~for $M\simeq \bbR^{m-1}\times \mathbb{T}$).
 Then by (\ref{noflat}) we obtain, for all $\varphi \in \EE^\mathrm{inv}$, 
 $\mathcal{O}_\varphi(\lambda_2) = \mathcal{O}_\varphi(\lambda_1) + \ip{\varphi_V}{\eta} = \mathcal{O}_\varphi(\lambda_1)$.
The origin of this pathology is the fact that $A_G$ is only an Abelian group and not a vector space (cf.~Corollary \ref{cor:AGgeneral}).
 Performing the quotient of the configuration space $\sect{M}{\mathcal{C}(\Xi)}$ by the gauge transformations
 that are of exponential form (that are all for $k=0$) we obtain again an affine space. However, performing the quotient of
 the resulting affine space by the Abelian group $A_G$ we obtain no affine space anymore
 (compare this with the quotient $\bbR/\mathbb{Z} \simeq \mathbb{T}$).
 The gauge invariant classical affine functionals $\{\mathcal{O}_\varphi: \varphi \in \EE^\mathrm{inv}\}$ 
 do not take into account the nontrivial topology of the quotient of the configuration space by the full gauge group.
  For this reason one should enlarge the algebra of  gauge invariant observables constructed in this paper
   to include additional elements  which can separate all gauge equivalence classes of connections. A natural candidate 
  are Wilson loops, but, being too singular objects localized on curves, they cannot be added easily to the present formalism
  used in algebraic quantum field theory. 
  We will come back to this issue in our future investigations.
\end{rem}
\sk

To conclude this section we characterize the radical $\mathcal{N}$ of the 
presymplectic vector space $(\EE,\tau)$ of Proposition
\ref{propo:presymp}.
\begin{theo}\label{theo:radicalexplicit}
Let $G\simeq \mathbb{T}^k\times\bbR^l$ be a connected Abelian Lie group,
$h$ a bi-invariant pseudo-Riemannian metric on $G$
and $\Xi=\big((M,\o,g,\t),(P,r)\big)$ an object in $G{-}\PrBu$.
Then the radical $\mathcal{N}$ of $(\EE,\tau )$ is
\begin{flalign}\label{eqn:radicalrefined}
\mathcal{N}=\big\{\psi \in  \EE^\mathrm{inv} : 
h^{-1}(\psi_V) \in \delta\Omega^2_{0,\dd}(M,i\bbR)^{\oplus_k}\oplus \delta\big(\Omega^2_0(M,\bbR)\cap \dd \Omega^1_\mathrm{tc}(M,\bbR)\big)^{\oplus_{l}}   \big\}/\MW^\ast\big[\Omega^1_0(M,\g^\ast)\big]~.
\end{flalign}
\end{theo}
\begin{proof}
Let $[\psi]$ be an element of the vector space on the right hand side of (\ref{eqn:radicalrefined}). Any representative $\psi$ is such that
$h^{-1}(\psi_V) = \delta \eta + \delta \dd \zeta$ for some $\eta \in \Omega^2_{0,\dd}(M,i\bbR)^{\oplus_k}$
and $\zeta\in \Omega^1_\mathrm{tc}(M,\bbR)^{\oplus_l}$.
By Theorem \ref{theo:gaugeinvspace} any $\varphi\in \EE^\mathrm{inv}$ is such that
$\varphi_V = \delta \alpha + \beta$ for some $\alpha\in \Omega^2_0(M,i\bbR)^{\oplus_k}$ and
$\beta\in \Omega^1_{0,\delta}(M,\bbR)^{\oplus_l}$. As a consequence,
\begin{flalign}
\nn \tau([\varphi],[\psi]) &= \ip{\varphi_V}{G_{(1)}\big(h^{-1}(\psi_V)\big)}
=\ip{\delta\alpha}{G_{(1)}(\delta\eta)} + \ip{\beta}{G_{(1)}(\delta\dd\zeta)}\\
&=\ip{\alpha}{\dd\delta G_{(2)}(\eta)} + \ip{\beta}{\delta\dd G_{(1)}(\zeta)}=
-\ip{\alpha}{\delta\dd G_{(2)}(\eta)} - \ip{\beta}{\dd\delta G_{(1)}(\zeta)}=0~,
\end{flalign}
hence the vector space on the right hand side of (\ref{eqn:radicalrefined}) is contained in the radical $\mathcal{N}$.
To show that it is equal to the radical let $\psi\in \EE^\mathrm{inv}$ be any element satisfying, for all $\varphi\in\EE^\mathrm{inv}$,
$\tau([\varphi],[\psi])=0$. Using again the decomposition $\varphi_V = \delta \alpha + \beta$ for some
$\alpha\in \Omega^2_0(M,i\bbR)^{\oplus_k}$ and $\beta\in \Omega^1_{0,\delta}(M,\bbR)^{\oplus_l}$,
as well as the decomposition $h^{-1}(\psi_V) = \delta \eta + \delta \epsilon$, where $\eta\in \Omega^2_{0,\dd}(M,i\bbR)^{\oplus_k}$
and $\epsilon\in \Omega^2_{0,\dd}(M,\bbR)^{\oplus_l}$ (which is possible due to Corollary \ref{cor:radicalinclusions}),
this condition yields
\begin{flalign}\label{eqn:tmpdualityradical}
0=\tau([\varphi],[\psi]) =
 \ip{\delta\alpha}{G_{(1)}(\delta\eta)} + \ip{\beta}{G_{(1)}(\delta\epsilon)} = \ip{\beta}{G_{(1)}(\delta\epsilon)}~.
\end{flalign}
By (\ref{eqn:tmpdualityradical}) and Poincar{\'e} duality
there exists a $\gamma\in C^\infty(M,\bbR)^{\oplus_l}$, such that $G_{(1)}(\delta\epsilon) = \dd \gamma$.
Applying the codifferential to this equation we find that $\gamma$ satisfies
 the wave equation $\delta\dd\gamma = \square_{(0)}(\gamma) =0$, hence by \cite{SDH12}  there exists a
 $\theta \in C^\infty_\mathrm{tc}(M,\bbR)^{\oplus_l}$ such that $\gamma = G_{(0)}(\theta) $.
 Plugging this into the equation above yields $G_{(1)}(\delta\epsilon) = \dd \gamma = G_{(1)}(\dd\theta)$, which implies
 $\delta\epsilon = \dd\theta + \square_{(1)}(\zeta)$ for some $\zeta\in \Omega^1_\mathrm{tc}(M,\bbR)^{\oplus_l}$.
 Applying $\dd$ and using that $\epsilon$ is closed we obtain $\epsilon = \dd\zeta$, which shows that any element in the radical is 
 contained in the vector space on the right hand side of (\ref{eqn:radicalrefined}).
\end{proof}


\section{\label{sec:phasespacefunctor}The phase space functor and $\CCR$-quantization}
In this section we show that the association of the presymplectic vector space $(\EE,\tau)$
in Proposition \ref{propo:presymp} to objects $\Xi=\big((M,\o,g,\t),(P,r)\big)$ in $G{-}\PrBu$ is functorial.
We are going to construct a covariant functor $\PhaseSpace : G{-}\PrBu \to \PreSymp$, where 
the latter category is that of presymplectic vector spaces with compatible morphisms, that are however not
assumed to be injective (see the definition below). We will then derive some important properties of the functor.
\begin{defi}
The category $\PreSymp$ consists of the following objects and morphisms:
\begin{itemize}
\item An object is a tuple $(\EE,\tau)$, where $\EE$ is a (possibly infinite dimensional)
vector space over $\bbR$ and $\tau: \EE\times\EE\to \bbR$ is an antisymmetric bilinear map (a presymplectic structure).
\item A morphism is a linear map $L:\EE_1\to \EE_2$ (not necessarily injective), which preserves the presymplectic structures,
i.e.~$\tau_2(L(v),L(w)) = \tau_1(v,w)$, for all $v,w\in \EE_1$.
\end{itemize}
\end{defi}
\sk

Before constructing the phase space functor $\PhaseSpace$ we spell out two lemmas characterizing
the compatibility of Maxwell's affine differential operator $\MW$, the Hodge-d'Alembert operators $\square_{(k)}$ 
and their Green's operators $G_{(k)}^\pm$ with morphisms in $G{-}\PrBu$.
\begin{lem}\label{lem:MWcomp}
Let $G$ be an Abelian Lie group and let  $f:\Xi_1\to \Xi_2$ be a morphism in $G{-}\PrBu$. 
Then the following diagram commutes
\begin{flalign}\label{eqn:MWcomp}
\xymatrix{
\ar[d]_-{f^\ast}\sect{M_2}{\mathcal{C}(\Xi_2)} \ar[rr]^-{\MW_2} & & \Omega^1(M_2,\g)\ar[d]^-{\underline{f}^\ast}\\
\sect{M_1}{\mathcal{C}(\Xi_1)} \ar[rr]^-{\MW_1} & & \Omega^1(M_1,\g)\\
}
\end{flalign}
where $f^\ast := \Gamma^\infty(\mathcal{C}(f))$ is defined in (\ref{eqn:connectionmorph}) and
$\underline{f}^\ast$ is the usual pull-back along the induced map $\underline{f}:M_1\to M_2$.
More abstractly, this entails that the Maxwell operator 
$\MW$ defines a natural transformation $\Gamma^\infty \circ \mathcal{C}\Rightarrow
\Omega^1_\mathrm{base}$.
\end{lem}
\begin{proof}
In Lemma \ref{lem:curvaffineop} 
and the text below we have observed that the curvature maps $\underline{\mathcal{F}}$ can be regarded
as a natural transformation $\Gamma^\infty \circ \mathcal{C}\Rightarrow
\Omega^2_\mathrm{base}$. Explicitly, we have that $\underline{\mathcal{F}}_1 \circ f^\ast
 = \underline{f}^\ast\circ \underline{\mathcal{F}}_2$.
Furthermore, using that by hypothesis $\underline{f}:M_1\to M_2$ is an isometric and orientation preserving embedding,
we obtain for the codifferentials $\delta_1 \circ\underline{f}^\ast = \underline{f}^\ast \circ \delta_2$.
This implies that $\MW_1 \circ f^\ast = \underline{f}^\ast \circ \MW_2$ and 
shows the commutativity of the diagram (\ref{eqn:MWcomp}).
\end{proof}
\begin{lem}\label{lem:HdAcomp}
Let $G$ be an Abelian Lie group and let $f:\Xi_1\to \Xi_2$ be a morphism in $G{-}\PrBu$. 
 \begin{itemize}
 \item[a)] The following diagram commutes for all $k$:
\begin{flalign}\label{eqn:HdAcomp}
\xymatrix{
\ar[d]_-{\underline{f}^\ast}\Omega^k(M_2,\g^\ast)\ar[rr]^-{\square_{2\,(k)}} & & \Omega^k(M_2,\g^\ast)\ar[d]^-{\underline{f}^\ast}\\
\Omega^k(M_1,\g^\ast)\ar[rr]^-{\square_{1\,(k)}} & & \Omega^k(M_1,\g^\ast)\\
}
\end{flalign}
More abstractly, this entails that the d'Alembert operators $\square_{(k)}$ 
define natural transformations $\Omega^k_\mathrm{base} \Rightarrow \Omega^k_\mathrm{base}$.
\item[b)] The Green's operators satisfy $G_{1\,(k)}^\pm = 
\underline{f}^\ast \circ G_{2\,(k)}^{\pm}\circ \underline{f}_\ast $, where
$\underline{f}_\ast $ denotes the push-forward of compactly supported forms along $\underline{f}:M_1\to M_2$.
\end{itemize}
\end{lem}
\begin{proof}
The commutative diagram (\ref{eqn:HdAcomp}) is a consequence both of 
$\square_{(k)} = \delta\circ \dd + \dd\circ \delta$ and of the fact that $\dd$ 
and $\delta$ are natural transformations, i.e.~$\dd:\Omega^k_\mathrm{base} \Rightarrow \Omega^{k+1}_\mathrm{base}$
and $\delta : \Omega^k_\mathrm{base}\Rightarrow \Omega^{k-1}_\mathrm{base}$.
\sk

To prove b) first notice that
$\underline{f}_\ast\big(\underline{f}^\ast(\eta)\big) =\eta$, for all
$\eta\in \Omega^k_0(\underline{f}[M_1],\g^\ast)\subseteq \Omega^k_0(M_2,\g^\ast)$, and that 
$\underline{f}^\ast\big(\underline{f}_\ast(\eta)\big) = \eta$, for all
$\eta\in \Omega^k_0(M_1,\g^\ast)$. 
Let us define $\widetilde{G}_{1\,(k)}^\pm := 
\underline{f}^\ast  \circ G_{2\,(k)}^{\pm}\circ \underline{f}_\ast$.
We show that $\widetilde{G}^\pm_{1\,(k)}$ are retarded/advanced Green's operators for $\square_{1\,(k)}$
and thus by uniqueness it follows the claim $\widetilde{G}^\pm_{1\,(k)} = G^\pm_{1\,(k)}$.
Due to the diagram (\ref{eqn:HdAcomp}) and the above properties of $\underline{f}_\ast$ and $\underline{f}^\ast$
we obtain 
\begin{subequations}
\begin{flalign}
 \square_{1\,(k)} \circ \widetilde{G}^\pm_{1\,(k)} = 
\square_{1\,(k)}\circ  \underline{f}^\ast \circ G_{2\,(k)}^{\pm}\circ \underline{f}_\ast 
=
\underline{f}^\ast \circ \square_{2\,(k)} \circ G_{2\,(k)}^{\pm}\circ \underline{f}_\ast  = \id_{\Omega_0^k(M_1,\g^\ast)}
\end{flalign}
and on $\Omega_0^k(M_1,\g^\ast)$
\begin{flalign}
\widetilde{G}^\pm_{1\,(k)} \circ \square_{1\,(k)} = \underline{f}^\ast \circ G_{2\,(k)}^{\pm}\circ \underline{f}_\ast \circ  \square_{1\,(k)} 
 =
 \underline{f}^\ast  \circ G_{2\,(k)}^{\pm}\circ \square_{2\,(k)} \circ \underline{f}_\ast = \id_{\Omega_0^k(M_1,\g^\ast)}~.
\end{flalign}
\end{subequations}
Thus, $\widetilde{G}^\pm_{1\,(k)}$ are Green's operators for $\square_{1\,(k)}$. They are retarded/advanced Green's operators,
since for all $\eta\in\Omega_0^k(M_1,\g^\ast)$,
\begin{flalign}
\supp\big(\widetilde{G}_{1\,(k)}^\pm(\eta)\big) 
\subseteq \underline{f}^{-1}\big[J_{M_2}^\pm\big(\underline{f}[\supp(\eta)]\big)\big]= J_{M_1}^\pm\big(\supp(\eta)\big)~,
\end{flalign}
where in the second step we have used that $\underline{f}[M_1]\subseteq M_2$ is by hypothesis causally compatible.
\end{proof}

\begin{defi}\label{defi:dualmap}
Let $G$ be an Abelian Lie group and let $f:\Xi_1\to \Xi_2$ be a morphism in $G{-}\PrBu$.
We define the linear map $f_\ast:\EE^\mathrm{kin}_1\to \EE^\mathrm{kin}_2$ by, 
for all $\varphi \in \EE^\mathrm{kin}_1$
and $\lambda \in \sect{M_2}{\mathcal{C}(\Xi_2)}$,
\begin{flalign}
\int_{M_2} \vol_2\,\big(f_\ast(\varphi)\big)(\lambda) = \int_{M_1}\vol_1\,\varphi\big(f^\ast(\lambda)\big)~.
\end{flalign}
\end{defi}
\sk

\begin{theo}\label{theo:phasespacefunc}
Let $G$ be an Abelian Lie group and $h$ a bi-invariant pseudo-Riemannian metric on $G$.
Then there exists a covariant functor $\PhaseSpace: G{-}\PrBu \to \PreSymp$. It associates to any object
$\Xi$ the presymplectic vector space $\PhaseSpace(\Xi) = (\EE,\tau)$
which has been constructed in Proposition \ref{propo:presymp}.
To a morphism $f:\Xi_1\to \Xi_2$ the functor associates the morphism in $\PreSymp$ given by
 \begin{flalign}\label{eqn:phasespacefunc}
 \PhaseSpace(f): \PhaseSpace(\Xi_1)\to \PhaseSpace(\Xi_2)~,~~[\varphi] \mapsto [f_\ast(\varphi)]~,
 \end{flalign}
 where the linear map $f_\ast$ is given in Definition \ref{defi:dualmap}.
\end{theo}
\begin{proof}
First, we show that $f_\ast$ maps $\EE^\mathrm{inv}_1$ to $\EE^\mathrm{inv}_2$.
For any $\varphi\in\EE^\mathrm{inv}_1$ the linear part satisfies
$f_\ast(\varphi)_V = \underline{f}_\ast(\varphi_V)$
and hence, for all $\widehat{g}\in C^\infty(M_2,G) $,
\begin{flalign}
\ip{f_\ast(\varphi)_V}{\widehat{g}^\ast(\mu_{G})}_2 = 
\ip{\varphi_V}{\underline{f}^\ast\big(\widehat{g}^\ast(\mu_{G})\big)}_1
=\ip{\varphi_V}{(\widehat{g}\circ\underline{f})^\ast(\mu_{G})\big)}_1
=0~,
\end{flalign}
which implies that $f_\ast(\varphi)\in \EE^\mathrm{inv}_2$.
\sk

Next, we prove that (\ref{eqn:phasespacefunc}) is well-defined, that is, for
all $\eta\in \Omega_0^1(M_1,\g^\ast)$ we have 
$f_\ast\big(\MW^\ast_1(\eta)\big) \in \MW_2^\ast\big[\Omega_0^1(M_2,\g^\ast)\big]$.
This property is a consequence of the following short calculation, for all $\lambda\in \sect{M_2}{\mathcal{C}(\Xi_2)}$,
\begin{flalign}
\nn \int_{M_2} \vol_2\,\Big(f_\ast\big(\MW^\ast_1(\eta)\big) \Big)(\lambda) & 
=\ip{\eta}{\MW_1\big(f^\ast(\lambda)\big)}_1 =\ip{\eta}{\underline{f}^\ast\big(\MW_2(\lambda)\big)}_1 \\
&= \ip{\underline{f}_\ast(\eta)}{\MW_2(\lambda)}_2 =\int_{M_2} \vol_2\,\Big(\MW_2^\ast\big(\underline{f}_\ast(\eta)\big)\Big)(\lambda)~,
\end{flalign}
where in the second equality we have used Lemma \ref{lem:MWcomp}.
\sk

It remains to be shown that the linear map $\PhaseSpace(f)$ in (\ref{eqn:phasespacefunc}) 
preserves the presymplectic structures.
Let us take two arbitrary $[\varphi],[\psi]\in \EE_1$. Then
\begin{flalign}
\tau_{2}\big([f_\ast(\varphi)],[f_\ast(\psi)]\big) = \ip{f_\ast(\varphi)_V}{ G_{2\,(1)}\big(f_\ast(\psi)_V\big)}_{2,\,h}~.
\end{flalign}
Using again that $f_\ast(\varphi)_V = \underline{f}_\ast (\varphi_V)$
(and similar for $\psi$) yields
\begin{flalign}
\nn \tau_{2}\big([f_\ast(\varphi)],[f_\ast(\psi)]\big) &= 
\ip{\underline{f}_\ast(\varphi_V)}{G_{2\,(1)}\big(\underline{f}_\ast(\psi_V)\big)}_{2,\,h}
=\ip{\varphi_V}{\big(\underline{f}^\ast\circ G_{2\,(1)}\circ \underline{f}_\ast \big)(\psi_V)}_{1,\,h}\\
&=\ip{\varphi_V}{G_{1\,(1)}(\psi_V)}_{1,\,h} = \tau_{1}([\varphi],[\psi])~.
\end{flalign}
In the third equality  we have used Lemma \ref{lem:HdAcomp} b).
\end{proof}
\begin{rem}\label{rem:noninjective}
The covariant functor $\PhaseSpace : G{-}\PrBu \to \PreSymp$ does not satisfy the locality property
stating that for any morphism $f:\Xi_1\to \Xi_2$ in $G{-}\PrBu$ the morphism 
$\PhaseSpace(f)$ is injective. We will show this failure first by using the simplest example $G=U(1)\simeq \mathbb{T}$
 and we refer to Section \ref{sec:chargezerofunctor} for a possible solution of this problem. 
Let $\Xi_2$ be an object in $U(1){-}\PrBu$ such that $(M_2,\o_2,g_2,\t_2)$ is the 
$m$-dimensional Minkowski spacetime ($m\geq 4$).
Let us denote by  $\Xi_1$ the object in $U(1){-}\PrBu$ that is obtained by restricting
all data of $\Xi_2$ to the causally compatible and globally hyperbolic open subset $M_1:= M_2\setminus J_{M_2}(\{0\})$,
where $\{0\}$ is the set of a single point in Minkowski spacetime (cf.~\cite[Lemma A.5.11]{Bar:2007zz}).
Notice that $M_1$ is diffeomorphic to $\bbR^{2}\times S^{m-2}$, where $S^{m-2}$ is the ${m{-}2}$-sphere.
The canonical embedding (via the identity) $f:\Xi_1\to \Xi_2$ is a morphism in $U(1){-}\PrBu$.
Let us take any nonexact element $\eta\in \Omega^2_{0,\dd}(M_1,\g^\ast)$, which exists,
since per Poincar{\'e} duality $H^{m-2}_\mathrm{dR}(M_1,\g) \simeq H^2_{0\,\mathrm{dR}}(M_1,\g^\ast)$
and $H^{m-2}_\mathrm{dR}(M_1,\g)\simeq \g\simeq i\,\bbR$ since $M_1$ is homotopy equivalent to $S^{m-2}$.
Applying the formal adjoint of the curvature affine differential operator we obtain a nontrivial element
$\big[\underline{\mathcal{F}_1}^\ast(\eta)\big] \in \PhaseSpace(\Xi_1)$ (this
element is contained in the radical $\mathcal{N}_1$, cf.~Theorem \ref{theo:radicalexplicit}).
Under the morphism $\PhaseSpace(f)$ we obtain
\begin{flalign}
\PhaseSpace(f)\big(\big[\underline{\mathcal{F}_1}^\ast(\eta)\big]\big) &=
\big[f_\ast\big(\underline{\mathcal{F}_1}^\ast(\eta)\big)\big] = \big[\underline{\mathcal{F}_2}^\ast\big(\underline{f}_\ast(\eta)\big)\big] =\big[\underline{\mathcal{F}_2}^\ast(\dd\xi)\big] = \big[\MW_2^\ast(\xi)\big]=0~.
\end{flalign}
In the third equality  we have used that $\underline{f}_\ast(\eta) 
\in \Omega^2_{0,\dd}(M_2,\g^\ast)$
is exact since $M_2$ is the Minkowski spacetime. By Remark \ref{rem:Nmin} the same conclusion holds true
for $G=\bbR$ and hence for generic $G\simeq \mathbb{T}^k\times\bbR^l$.
\end{rem}

\begin{theo}\label{theo:classcausal}
The covariant functor $\PhaseSpace: G{-}\PrBu\to\PreSymp$ satisfies the classical causality property:\sk

Let $f_i:\Xi_i\to \Xi_3$, $i=1,2$, be two morphisms in $G{-}\PrBu$, such that
$\underline{f_1}[M_1]$ and $\underline{f_2}[M_2]$ are causally disjoint in $M_3$.
Then $\tau_{3}$ acts trivially among the vector subspaces
$\PhaseSpace(f_1)\big[\PhaseSpace(\Xi_1)\big]$ and 
$\PhaseSpace(f_2)\big[\PhaseSpace(\Xi_2)\big]$ of  
$\PhaseSpace(\Xi_3)$. That is, for all $[\varphi]\in \PhaseSpace(\Xi_1)$ and $[\psi]\in \PhaseSpace(\Xi_2)$,
\begin{flalign}
\tau_{3}\big(\PhaseSpace(f_1)([\varphi]),\PhaseSpace(f_2)([\psi])\big) =0~.
\end{flalign}
\end{theo}
\begin{proof}
From (\ref{eqn:phasespacefunc}) and (\ref{eqn:presymp}) it follows that
\begin{flalign}
\tau_{3} \big(\PhaseSpace(f_1)([\varphi]),\PhaseSpace(f_2)([\psi])\big) = 
\ip{\underline{f_1}_\ast(\varphi_V)}{G_{3\,(1)}\big(\underline{f_2}_\ast(\psi_V)\big)}_{3,\,h} =0~,
\end{flalign}
since the supports $\supp\big(\underline{f_1}_\ast(\varphi_V)\big)\subseteq \underline{f_1}[M_1]$
 and $ \supp\big( G_{3\,(1)}\big(\underline{f_2}_\ast(\psi_V)\big)\big)\subseteq
  J_{M_3}(\underline{f_2}[M_2])$ are by hypothesis disjoint.
\end{proof}
\begin{theo}\label{theo:classtimeslice}
The covariant functor $\PhaseSpace: G{-}\PrBu\to\PreSymp$ satisfies the classical time-slice axiom:\sk

Let  $f:\Xi_1\to \Xi_2$ a morphism in $G{-}\PrBu$,
 such that $\underline{f}[M_1]\subseteq M_2$ contains a Cauchy surface of $M_2$.
Then
\begin{flalign}
\PhaseSpace(f): \PhaseSpace(\Xi_1)\to \PhaseSpace(\Xi_2)
\end{flalign}
is an isomorphism.
\end{theo}
\begin{proof}
Let us define $\Xi_2\vert_{\underline{f}[M_1]} := \big((\underline{f}[M_1],\o_2\vert_{\underline{f}[M_1]},g_2\vert_{\underline{f}[M_1]},\t_2\vert_{\underline{f}[M_1]}),(P_2\vert_{\underline{f}[M_1]},r_2\vert_{\underline{f}[M_1]})\big)$, where
 $P_2\vert_{\underline{f}[M_1]}$ denotes the restriction of the principal $G$-bundle $(P_2,r_2)$ over $M_2$
to $\underline{f}[M_1]\subseteq M_2$. Notice that $\Xi_2\vert_{\underline{f}[M_1]}$ is an object in $G{-}\PrBu$ and
by definition of the morphisms in this category, $f: \Xi_1\to \Xi_2\vert_{\underline{f}[M_1]} $
is an isomorphism.
As a consequence of functoriality, we obtain an isomorphism in $\PreSymp$
\begin{flalign}
\PhaseSpace(f): \PhaseSpace(\Xi_1) \to \PhaseSpace(\Xi_2\vert_{\underline{f}[M_1]})~.
\end{flalign}
Hence, the proof would follow if we could show that 
the canonical map $ \PhaseSpace(\Xi_2\vert_{\underline{f}[M_1]}) \to  \PhaseSpace(\Xi_2)$, 
$[\varphi]\mapsto [\varphi]$ is an isomorphism under the hypotheses of this theorem. 
\sk

Let us first prove injectivity of the canonical map: Let $[\varphi]\in \PhaseSpace(\Xi_2\vert_{\underline{f}[M_1]})$ be 
such that when interpreted via the canonical map as an element in $\PhaseSpace(\Xi_2)$ we have $[\varphi]=0$. As a consequence,
$[\varphi] \in \PhaseSpace(\Xi_2\vert_{\underline{f}[M_1]})$ has to be in the radical $\mathcal{N}_2\vert_{\underline{f}[M_1]}$
and by Corollary \ref{cor:radicalinclusions} there exists for any representative $\varphi$ an 
$\eta\in \Omega^2_{0,\dd}(\underline{f}[M_1],\g^\ast)$ such that $\varphi_V = \delta_2\eta$.
Notice that due to the quotient in Corollary \ref{cor:radicalinclusions} the equivalence class $[\varphi]$
only depends on the cohomology class $[\eta]\in H^{2}_{0\,\mathrm{dR}}(\underline{f}[M_1],\g^\ast)$.
By a theorem of Bernal and S{\'a}nchez \cite{Bernal:2004gm} and by the hypothesis that $\underline{f}[M_1]$ 
contains a Cauchy surface of $M_2$ we have that $\underline{f}[M_1]$ and $M_2$ are homotopy equivalent 
(notice also that $\mathrm{dim}(\underline{f}[M_1]) = \mathrm{dim}(M_2)$). This entails that the inclusion 
$i:\underline{f}[M_1]\to M_2$ induces an isomorphism $i^\ast$ between the corresponding de Rham cohomology groups. 
In particular, for each $[\omega]\in H^{\mathrm{dim}(M_2)-2}_\mathrm{dR}(\underline{f}[M_1],\g)$, 
there exists $[\omega^\prime]\in H^{\mathrm{dim}(M_2)-2}_\mathrm{dR}(M_2,\g)$ 
such that $i^\ast([\omega^\prime])=[\omega]$. We use this isomorphism to show that 
$[\eta]\in H^{2}_{0\,\mathrm{dR}}(\underline{f}[M_1],\g^\ast)$ is trivial. 
This fact follows from Poincar{\'e} duality if the pairing between 
$[\eta]\in H^{2}_{0\,\mathrm{dR}}(\underline{f}[M_1],\g^\ast)$ 
and any $[\omega]\in H^{\mathrm{dim}(M_2)-2}_\mathrm{dR}(\underline{f}[M_1],\g)$ gives zero. 
Such pairing is expressed in terms of the integral $\int_{\underline{f}[M_1]}[\eta]\wedge[\omega]$. 
Using the argument above, we can replace $[\omega]$ with $i^\ast([\omega^\prime])$ in the last formula 
and compute explicitly the integral with an arbitrary choice of representatives. We obtain that
\begin{flalign}
\int_{\underline{f}[M_1]}[\eta]\wedge[\omega]=\int_{\underline{f}[M_1]}\eta\wedge i^\ast(\omega^\prime)
=\int_{M_2}i_\ast(\eta)\wedge\omega=\int_{M_2}[\eta]\wedge[\omega^\prime]=0~,
\end{flalign}
where the last $[\eta]$ is the trivial element of $H_{0\,\mathrm{dR}}^2(M_2,\g^\ast)$, 
since $[\varphi]=0$ when regarded in $\PhaseSpace(\Xi_2)$ per hypothesis. 
Thus, we can find a representative $\varphi$ of the class
$[\varphi] \in \PhaseSpace(\Xi_2\vert_{\underline{f}[M_1]})$ such that $\varphi_V =0$, i.e.~$\varphi = a\,\1_2$ 
with $a\in C_0^\infty(\underline{f}[M_1])$. Since $[\varphi]$ lies in the kernel of the canonical map
we obtain $0=\int_{M_2} \vol_2 \,a =\int_{\underline{f}[M_1]} \vol_2\,a  $
and thus $[\varphi] =0$ in $\PhaseSpace(\Xi_2\vert_{\underline{f}[M_1]})$. 
\sk

Next, we prove surjectivity of the canonical map: Let $[\varphi]\in \PhaseSpace(\Xi_2)$ be arbitrary and let $\varphi$
be any representative. Per hypothesis, there exists a Cauchy surface $\Sigma_2$ in $M_2$ which is contained in 
$\underline{f}[M_1]$. Then $\Sigma_1 := \underline{f}^{-1}[\Sigma_2]$ is a Cauchy surface in $M_1$, since
$\underline{f}:M_1\to \underline{f}[M_1]$ is an isometry. Let us choose two other Cauchy surfaces
$\Sigma^\pm_1$ with $\Sigma^\pm_1\cap \Sigma_1=\emptyset$ in the future/past of $\Sigma_1$ and let us denote
by $\Sigma_2^\pm := \underline{f}[\Sigma_1^\pm]$ their images, which are Cauchy surfaces in $M_2$ since
$\underline{f}[M_1]$ is causally compatible. Let $\chi^+\in C^\infty(M_2)$ be any function such that
$\chi^+ \equiv 1 $ on $J_{M_2}^+(\Sigma_2^+)$ and $\chi^+ \equiv 0$ on $J_{M_2}^-(\Sigma_2^-)$.
We define $\chi^-\in C^\infty(M_2)$ by $\chi^+ + \chi^- \equiv 1$ on $M_2$.
Then $\eta:=\chi^+\,G^-_{(1)}(\varphi_V) + \chi^-\,G^{+}_{(1)}(\varphi_V) \in \Omega^1_0(M_2,\g^\ast)$ is of compact support
and the linear part of $\varphi^\prime := \varphi + \MW^\ast_2(\eta)$, given by $\varphi^\prime_V =\varphi_V - \delta_2\dd_2 \eta$,
vanishes outside of $\underline{f}[M_1]$ (remember that by Lemma \ref{lem:coclosed} $\delta_2\varphi_V=0$).
The constant affine part of $\varphi^\prime$ can be treated as in \cite[Theorem 5.6]{Benini:2012vi} by adding a suitable
element of $\mathrm{Triv}_2$ to $\varphi^\prime$, which leads to a representative $\varphi^{\prime\prime}$ of the same
class $[\varphi]$ that has compact support in $\underline{f}[M_1]$. The class
 $[\varphi^{\prime\prime}]\in \PhaseSpace(\Xi_2\vert_{\underline{f}[M_1]})$ proves  
 surjectivity of the canonical map.
\end{proof}
\sk

We quantize our theory by using the $\CCR$-functor, which we are 
going to briefly review to be self-contained.
Let us define the category $\astAlg$:
An object  is a  unital $\ast$-algebra  $\mathcal{A}$ over $\bbC$.
A morphism is a unital $\ast$-algebra homomorphism $\kappa: \mathcal{A}_1 \to \mathcal{A}_2$
 (not necessarily injective).
The $\CCR$-functor is the covariant functor $\CCR:\PreSymp \to \astAlg$
which associates to any object $(\EE,\tau)$ the 
 unital $\ast$-algebra $\CCR(\EE,\tau)= \mathcal{T}(\EE)/\mathcal{I}(\EE,\tau)$.
 $\mathcal{T}(\EE)$ is the complex tensor algebra over $\EE$ and $\mathcal{I}(\EE,\tau)$ is the two-sided ideal generated by
the elements $v \otimes_\bbC w - w \otimes_\bbC v -i\,\tau(v,w)\,\oone$, for all $v,w\in \EE$.
To any morphism $L:(\EE_1,\tau_1)\to (\EE_2,\tau_2)$ in $\PreSymp$ the functor
associates the morphism $\CCR(L)$ in $\astAlg$ which is defined on the tensor algebra by
$\CCR(L)\big(v_1\otimes_\bbC \cdots\otimes_\bbC v_k \big)= L(v_1)\otimes_\bbC \cdots \otimes_\bbC L(v_k)$,
for all $k\geq 1$ and $v_1,\dots,v_k\in \EE_1$. Since $L$ preserves the presymplectic structures,
this unital $\ast$-algebra homomorphism canonically induces to the quotients.

Using the same arguments as in \cite[Theorem 6.3]{Benini:2012vi} it follows immediately from
Theorem \ref{theo:classcausal} and Theorem \ref{theo:classtimeslice} the following
\begin{theo}\label{theo:QFTfunctor}
The covariant functor $\QFT := \CCR \circ \PhaseSpace: G{-}\PrBu\to\astAlg$ satisfies:

\begin{itemize}
\item[(i)] The quantum causality property:\sk

Let $f_i :\Xi_i\to \Xi_3$, $i=1,2$, be two morphisms
in $G{-}\PrBu$, such that $\underline{f_1}[M_1]$ and $\underline{f_2}[M_2]$
are causally disjoint in $M_3$. Then $\QFT(f_1)\big[\QFT(\Xi_1)\big]$ and
$\QFT(f_2)\big[\QFT(\Xi_2)\big]$ commute as subalgebras of $\QFT(\Xi_3)$.

\item[(ii)] The quantum time-slice axiom:\sk

Let $f:\Xi_1\to \Xi_2$ a morphism in $G{-}\PrBu$,
 such that $\underline{f}[M_1]\subseteq M_2$ contains a Cauchy surface of $M_2$.
Then
\begin{flalign}
\QFT(f): \QFT(\Xi_1) \to \QFT(\Xi_2)
\end{flalign} 
is an isomorphism.
 \end{itemize}
\end{theo}


\section{\label{sec:nattrafo}Generally covariant topological quantum fields}
According to \cite{Brunetti:2001dx}, a locally covariant quantum field is a natural transformation
from a covariant functor describing test sections to the covariant functor $\QFT$. In this section we 
introduce the concept of generally covariant topological quantum fields, that are natural transformations
from a covariant functor describing topological information to the functor $\QFT$,
and construct two examples which can be interpreted as magnetic and electric charge.
We have added the attribute `generally covariant' in `generally covariant topological quantum field' 
in order to distinguish it from the usual notion of topological quantum field theory \cite{Atiyah:1989vu}.
To simplify the discussion, in this section we restrict ourselves to the case $G=U(1)$.
We define the category $\mathsf{Vec}$ as follows: 
An object  is a (possibly infinite dimensional) vector space $V$ over $\bbR$.
A morphism is a linear map $L:V_1\to V_2$ (not necessarily injective).
\sk

Composing $\QFT:U(1){-}\PrBu \to \astAlg$ with the forgetful functor from from $\astAlg$ to $\mathsf{Vec}$ 
we can consider $\QFT$ as a covariant functor from $U(1){-}\PrBu$ to $\mathsf{Vec}$ 
(with a slight abuse of notation we denote
this covariant functor again by $\QFT$). The other covariant functors
from $U(1){-}\PrBu$ to $\mathsf{Vec}$ which enter our construction of generally covariant topological quantum fields 
are those of smooth singular homology with coefficients in the real vector space $\g^\ast = i\,\bbR$
(since the smooth and continuous singular homology are isomorphic, the smooth singular homology only contains
topological information).
The covariant functor $\mathfrak{H}_p:U(1){-}\PrBu \to \mathsf{Vec}$ is defined as follows: To any object $\Xi$
 the functor associates the $p$-th singular homology group of the base space
$\mathfrak{H}_p(\Xi) = H_p(M,\g^\ast)$.
To a morphism $f:\Xi_1\to \Xi_2$  the functor associates the push-forward of $p$-simplices, i.e.~
\begin{flalign}
\mathfrak{H}_p(f) : \mathfrak{H}_p(\Xi_1) \to \mathfrak{H}_p(\Xi_2)~,~~ \Big[\sum a_j \,\sigma_j\Big] 
\mapsto \Big[\sum a_j \, (\underline{f}\circ \sigma_j)\Big]  ~,
\end{flalign}
where $\sigma_j: \Delta^p \to M$ are $p$-simplices and $a_j\in\g^\ast$ are coefficients.
The singular cohomology is defined by duality, 
$H^\ast(M,\g) := \Hom_\bbR(H_\ast(M,\g^\ast),\bbR)$. Furthermore, by de Rham's theorem
there exists a vector space isomorphism 
$\mathcal{J}: H^p_{\mathrm{dR}}(M,\g) \to H^p(M,\g)\,,~[\eta] \mapsto \mathcal{J}([\eta])$,
where $\mathcal{J}([\eta])$ is the linear functional on $H_p(M,\g^\ast)$ defined by, for all $\sum a_j\,\sigma_j$,
\begin{flalign}
\mathcal{J}([\eta])\Big(\Big[\sum a_j\,\sigma_j\Big]\Big) 
= \sum a_j \int_{\Delta^p} \sigma_j^\ast(\eta)~,
\end{flalign}
where $\sigma_j^\ast$ is the pull-back of $\sigma_j:\Delta^p\to M$ and the duality pairing between
$\g^\ast$ and $\g$ is suppressed. By Poincar{\'e} duality there also exists
a vector space isomorphism $\mathcal{K}:  H_p(M,\g^\ast)\to H^{p}_{0\,\mathrm{dR}^\ast}(M,\g^\ast) $ (by the subscript 
$\mathrm{dR}^\ast$ we denote the cohomology groups of the codifferential $\delta$) specified by,
for all $[\sigma]\in H_p(M,\g^\ast)$ and $[\eta]\in H_{\mathrm{dR}}^p(M,\g)$,
\begin{flalign}\label{eqn:poincaresingtmp}
\ip{\mathcal{K}([\sigma])}{[\eta]} =\mathcal{J}([\eta])([\sigma])~.
\end{flalign}
The pairing $\ip{~}{~}:H^{p}_{0\,\mathrm{dR}^\ast}(M,\g^\ast) \times H_{\mathrm{dR}}^p(M,\g) \to \bbR$
on the left hand side is that induced by the pairing $\ip{\zeta}{\eta} = \int_M \zeta\wedge \ast(\eta)$ 
 of $p$-forms $\zeta\in \Omega^p_0(M,\g^\ast)$ and $\eta\in\Omega^p(M,\g)$.
\sk

We now can construct our first example of a generally covariant topological quantum field,
which by Remark \ref{rem:magneticinterpreation} below should be interpreted as magnetic charge (Chern class).
\begin{theo}\label{theo:magnattrafo}
Consider the two covariant functors $\mathfrak{H}_2,\QFT: U(1){-}\PrBu\to \mathsf{Vec}$.
 We associate to any object $\Xi$ 
 the morphism in $\mathsf{Vec}$ given by
 \begin{flalign}\label{eqn:magnaturaltrafo}
 \Psi^{\mathrm{mag}}_\Xi : \mathfrak{H}_2(\Xi)\to \QFT(\Xi)~,~~[\sigma] 
 \mapsto \big[ \underline{\mathcal{F}}^\ast(\mathcal{K}([\sigma])) \big]~,
 \end{flalign}
 where $\underline{\mathcal{F}}^\ast : \Omega^2_0(M,\g^\ast) \to \EE^\mathrm{kin}$ is the formal
 adjoint of the curvature affine differential operator.
 The collection $\Psi^\mathrm{mag} = \{\Psi^\mathrm{mag}_\Xi\}$ is a natural transformation
 from $\mathfrak{H}_2$ to $\QFT$.
\end{theo}
\begin{proof}
The map (\ref{eqn:magnaturaltrafo})
is well-defined due to the dual of the (Abelian) Bianchi identity
$\dd\circ \underline{\mathcal{F}} =0$. Furthermore, since any representative of the class 
$\mathcal{K}([\sigma])$
is coclosed, the linear part of $\underline{\mathcal{F}}^\ast(\mathcal{K}([\sigma]) )$ vanishes.
Hence, $\underline{\mathcal{F}}^\ast( \mathcal{K}([\sigma])) \in \EE^{\mathrm{inv}}$ 
is a representative of an element in the radical $\mathcal{N}$
and the image of (\ref{eqn:magnaturaltrafo}) is contained in $\EE\subseteq \QFT(\Xi) $.

Let $f:\Xi_1\to \Xi_2$ be a morphism in $U(1){-}\PrBu$. As a consequence of
$\underline{\mathcal{F}}$ being a natural transformation and 
$\mathcal{K}$ being a natural isomorphism we obtain that the following diagram commutes:
\begin{flalign}
\xymatrix{
\ar[d]_-{\mathfrak{H}_2(f)}\mathfrak{H}_2(\Xi_1) \ar[rr]^-{\Psi^{\mathrm{mag}}_{\Xi_1}} & & \QFT(\Xi_1)\ar[d]^-{\QFT(f)}\\
\mathfrak{H}_2(\Xi_2) \ar[rr]^-{\Psi^{\mathrm{mag}}_{\Xi_2}} & & \QFT(\Xi_2)
}
\end{flalign}
This proves that $\Psi^\mathrm{mag} = \{\Psi^\mathrm{mag}_\Xi\}$ is a natural transformation.
\end{proof}
\begin{rem}\label{rem:magneticinterpreation}
The interpretation of the natural transformation $\Psi^\mathrm{mag}$ is as follows:
When evaluated on any $\lambda\in \sect{M}{\mathcal{C}(\Xi)}$, 
the classical affine functional (\ref{eqn:affinefunctional}) corresponding to
$\underline{\mathcal{F}}^\ast(\mathcal{K}([\sigma]) )$ yields 
\begin{flalign}
\mathcal{O}_{\underline{\mathcal{F}}^\ast(\mathcal{K}([\sigma]) )}(\lambda) 
= \ip{\mathcal{K}([\sigma]) }{\underline{\mathcal{F}}(\lambda)}
=\sum a_j\,\int_{\Delta^2} \sigma_j^\ast\big(\underline{\mathcal{F}}(\lambda)\big)~.
\end{flalign}
Via this identification the elements in the image of the map 
$\Psi^\mathrm{mag}_\Xi$ determine the cohomology class $[\underline{\mathcal{F}}(\lambda)]\in H^2_\mathrm{dR}(M,\g)$
and hence the Chern class of the principal $U(1)$-bundle. In physics $[\underline{\mathcal{F}}(\lambda)]$
is called the magnetic charge. This is a purely topological information, which justifies our nomenclature: generally covariant topological quantum field. After $\CCR$-quantization, we should interpret the image of the map (\ref{eqn:magnaturaltrafo}) as magnetic 
charge observables, which can be assigned coherently to all objects in $U(1){-}\PrBu$, since $\Psi^\mathrm{mag}$ is a natural transformation. We note that the image of the map (\ref{eqn:magnaturaltrafo}) lies in the center of the algebra $\QFT(\Xi)$, hence
magnetic charge observables are not subject to Heisenberg's uncertainty relation and can be measured without quantum fluctuations.
\end{rem}
\sk

Motivated by \cite{SDH12} we will now construct a generally covariant topological quantum field,
which, on account of Remark \ref{rem:elinterpreation}, should be interpreted as electric charge.
For this we require a covariant functor which associates to any object $\Xi$ in $U(1){-}\PrBu$
the singular homology group $H_{\mathrm{dim}(M)-2}(M,\g^\ast)\simeq H^{\dim(M)-2}_{0\,\mathrm{dR}^\ast}(M,\g^\ast)$.
 This functor  exists
since the set of morphisms $\{f:\Xi_1\to \Xi_2\}$ is only nonempty between objects $\Xi_1$ and $\Xi_2$
where $M_1$ and $M_2$ have the same dimension (cf.~Definition \ref{def:prbucat}). We shall denote this covariant functor
by $\mathfrak{H}_{-2}:U(1){-}\PrBu \to \mathsf{Vec}$.
\begin{theo}\label{theo:eltrafo}
Consider the two covariant functors $\mathfrak{H}_{-2},\QFT : U(1){-}\PrBu \to \mathsf{Vec}$.
 We associate to any object $\Xi$ 
 the morphism in $\mathsf{Vec}$ given by
 \begin{flalign}\label{eqn:elnaturaltrafo}
 \Psi^{\mathrm{el}}_\Xi : \mathfrak{H}_{-2}(\Xi)\to \QFT(\Xi)~,~~[\sigma]
 \mapsto \big[ \underline{\mathcal{F}}^\ast\big({\ast}( \mathcal{K}([\sigma]) ) \big) \big]~.
 \end{flalign}
 The collection $\Psi^\mathrm{el} = \{\Psi^\mathrm{el}_\Xi\}$ is a natural transformation
 from $\mathfrak{H}_{-2}$ to $\QFT$.
\end{theo}
\begin{proof}
The map (\ref{eqn:elnaturaltrafo}) is well-defined, since for all $\chi \in \Omega_0^{\dim(M) -1}(M,\g^\ast)$,
$\underline{\mathcal{F}}^\ast\big({\ast{(\delta\chi)}}\big) = \MW^\ast(\ast(\chi))$ yields the trivial class in $\EE\subseteq \QFT(\Xi)$.
For any $\eta\in \Omega_{0,\delta}^{\dim(M)-2}(M,\g^\ast)$
the linear part of $\underline{\mathcal{F}}^\ast\big({\ast(\eta)}\big)$ is
$\underline{\mathcal{F}}^\ast\big({\ast(\eta)}\big)_V = \delta \,{\ast(\eta)}$, with $\ast(\eta)\in \Omega^2_{0,\dd}(M,\g^\ast)$.
Hence, $\underline{\mathcal{F}}^\ast\big({\ast}( \mathcal{K}([\sigma]) ) \big) \in \EE^{\mathrm{inv}}$ 
is a representative of an element in the radical $\mathcal{N}$
and the image of (\ref{eqn:elnaturaltrafo}) is contained in $\EE\subseteq \QFT(\Xi) $.
\sk

Let $f:\Xi_1\to \Xi_2$ be a morphism in $U(1){-}\PrBu$. 
Using that the Hodge operator is a natural isomorphism
and the same arguments as in the proof of Theorem
\ref{theo:magnattrafo} we obtain that the following diagram commutes:
\begin{flalign}
\xymatrix{
\ar[d]_-{\mathfrak{H}_{-2}(f)}\mathfrak{H}_{-2}(\Xi_1) \ar[rr]^-{\Psi^{\mathrm{el}}_{\Xi_1}} & & \QFT(\Xi_1)\ar[d]^-{\QFT(f)}\\
\mathfrak{H}_{-2}(\Xi_2) \ar[rr]^-{\Psi^{\mathrm{el}}_{\Xi_2}} & & \QFT(\Xi_2)
}
\end{flalign}
This proves that $\Psi^\mathrm{el} = \{\Psi^\mathrm{el}_\Xi\}$ is a natural transformation.
\end{proof}
\begin{rem}\label{rem:elinterpreation}
Following Remark \ref{rem:magneticinterpreation} we can interpret $\Psi^\mathrm{el} $ as a coherent assignment of
electric charge observables:
When evaluated on any solution $\lambda\in \sect{M}{\mathcal{C}(\Xi)}$ 
of the equation of motion $\MW(\lambda)=0$,
the classical affine functional (\ref{eqn:affinefunctional}) corresponding to
$\underline{\mathcal{F}}^\ast\big({\ast}( \mathcal{K}([\sigma]) ) \big) $  
yields
\begin{flalign}
\mathcal{O}_{\underline{\mathcal{F}}^\ast({\ast}( \mathcal{K}([\sigma]) ) ) }(\lambda) 
= \ip{\mathcal{K}([\sigma])}{\ast{\big(\underline{\mathcal{F}}(\lambda)\big)}}
=\sum a_j\,\int_{\Delta^{\dim(M)-2}} \sigma_j^\ast\big({\ast{\big(\underline{\mathcal{F}}(\lambda)\big)}}\big)~.
\end{flalign}
Via this identification the elements in the image of the map $\Psi^\mathrm{el}_\Xi$ determine the cohomology class 
$[{\ast{(\underline{\mathcal{F}}(\lambda))}}]\in H^{\dim(M)-2}_\mathrm{dR}(M,\g)$ that, via
 Gauss' law, is the electric charge. As in the previous case, 
the image of the map (\ref{eqn:elnaturaltrafo}) lies in the center of
the algebra $\QFT(\Xi)$, meaning that electric charge observables in the quantum theory
 are not subject to Heisenberg's uncertainty relation and can be measured without quantum fluctuations.
\end{rem}


\section{\label{sec:chargezerofunctor}The charge-zero functor and the locality property}
In the previous section we identified electric and magnetic charge observables in
the algebra $\QFT(\Xi) = \CCR\big(\PhaseSpace(\Xi)\big)$ for any object $\Xi$ in $U(1){-}\PrBu$. 
While magnetic charge observables are certainly very welcome in our framework since they can measure the
topology of the principal bundle, electric charges play a different role. By construction,
the covariant functor $\QFT:U(1){-}\PrBu\to \astAlg$ models quantized principal $U(1)$-connections
without the presence of any charged fields. As a consequence, all electric charge measurements should yield
 zero.
We are going to implement this physical feature into our framework by performing a different quotient
in the presymplectic vector spaces $(\EE,\tau)$ of Proposition \ref{propo:presymp}.
This extends \cite[Remark 4.17]{SDH12} to our principal bundle setting and leads to one possible 
solution of the locality problem\footnote{
We are grateful to Jochen Zahn for communicating to us an alternative strategy for solving the
locality problem.
}. 
It is then rather straightforward to show that there is a covariant functor $\PhaseSpace^{0}: U(1){-}\PrBu \to \PreSymp$,
the charge-zero phase space functor, which associates these  presymplectic  vector spaces
to objects in $U(1){-}\PrBu$.
Interestingly, the functor $\PhaseSpace^{0}$ satisfies, in addition to the classical causality property and the classical time-slice axiom,
the locality property stating that for any morphism $f$ in $U(1){-}\PrBu$
the morphism $\PhaseSpace^{0}(f)$ in $\PreSymp$ is injective. Due to Remark \ref{rem:noninjective}
this is not the case for the functor $\PhaseSpace$ constructed in Section \ref{sec:phasespacefunctor}.
Composing the charge-zero phase space functor with the $\CCR$-functor we obtain a covariant functor
$\QFT^{0}$ that satisfies all axioms of locally covariant quantum field theory, i.e.~the quantum causality property, 
the quantum time-slice axiom and injectivity of $\QFT^{0}(f)$ for any morphism $f$ in $U(1){-}\PrBu$.
\sk

An interesting problem would be to understand if our physically well-motivated, however
slightly ad hoc, procedure of identifying the electric charges with zero can be explained
within the formalism developed by Fewster \cite{Fewster:2012yc}.\footnote{
We are grateful to the anonymous referee for pointing this out to us.} 
The basic idea of this paper is
to identify the group of automorphisms of a quantum field theory functor 
with the `global gauge group' of the theory. This group then can be used
to characterize the invariant subalgebras in $\QFT(\Xi)$, which should be the true observables
of the theory. Applied to our setting, this idea might provide a possibility to interpret the charge-zero algebras
$\QFT^0(\Xi)$ as arising from those subalgebras of $\QFT(\Xi)$ which are invariant under the automorphism group.
However, concrete statements on this relation require a computation of the automorphism
group of the functor $\QFT$, which is rather technical and beyond the scope of this article. We hope to
come back to this issue in a future work.

\sk

Let $\Xi=\big((M,\o,g,\t),(P,r)\big)$ be an object in $U(1){-}\PrBu$
and $\EE^\mathrm{inv}$ the vector space characterized in Theorem \ref{theo:gaugeinvspace}.
Notice that the vector subspace $\underline{\mathcal{F}}^\ast\big[\Omega^2_{0,\dd}(M,\g^\ast)\big]\subseteq \EE^\mathrm{inv}$
contains $\MW^\ast\big[\Omega^1_0(M,\g^\ast)\big]$ as a vector subspace as well as the electric charge observables
of Theorem \ref{theo:eltrafo}. Hence, by considering the quotient 
$\EE^0:= \EE^\mathrm{inv}/\underline{\mathcal{F}}^\ast\big[\Omega^2_{0,\dd}(M,\g^\ast)\big]$
we implement the equation of motion and identify all  electric charges with zero.
\begin{lem}\label{lem:presympzero}
Let $\Xi$ be an object in $U(1){-}\PrBu$ and $h$ any bi-invariant pseudo-Riemannian metric on $U(1)$.
\begin{itemize}
\item[a)] $\EE^{0}:= \EE^\mathrm{inv}/\underline{\mathcal{F}}^\ast\big[\Omega^2_{0,\dd}(M,\g^\ast)\big] $ can be equipped
with the presymplectic structure
\begin{flalign}\label{eqn:presympzero}
\tau^0: \EE^0\times\EE^0 \to \bbR~,~~\big([\varphi],[\psi]\big)\mapsto \tau^0\big([\varphi],[\psi]\big) = \ip{\varphi_V}{G_{(1)}(\psi_V)}_h~.
\end{flalign}
In other words, $(\EE^{0} ,\tau^0)$ is a presymplectic vector space.
\item[b)] The radical $\mathcal{N}^{0}$ of $(\EE^{0} ,\tau^0)$ is
\begin{flalign}
\mathcal{N}^0 = \big[\big\{\varphi\in \EE^\mathrm{inv} : 
\varphi_V =0\big\}\big]~.
\end{flalign}
\end{itemize}
\end{lem}
\begin{proof}
This is a direct consequence of Theorem \ref{theo:radicalexplicit}. 
\end{proof}
Similar to Theorem \ref{theo:phasespacefunc} we obtain that the association of these presymplectic vector spaces is functorial.
\begin{theo}\label{theo:phasespacefunczero}
Let $h$ be a bi-invariant pseudo-Riemannian metric on $U(1)$.
Then there exists a covariant functor $\PhaseSpace^{0}: U(1){-}\PrBu \to \PreSymp$. It associates to any object
$\Xi$  the presymplectic vector space $\PhaseSpace^0(\Xi) = (\EE^0,\tau^0)$
 which has been constructed in Lemma \ref{lem:presympzero}.
To a morphism $f:\Xi_1\to \Xi_2$  the functor associates the morphism in $\PreSymp$ given by
 \begin{flalign}\label{eqn:phasespacefunczero}
 \PhaseSpace^0(f): \PhaseSpace^0(\Xi_1)\to \PhaseSpace^0(\Xi_2)~,~~[\varphi] \mapsto [f_\ast(\varphi)]~,
 \end{flalign}
 where the linear map $f_\ast$ is given in Definition \ref{defi:dualmap}.
\end{theo}
\begin{proof}
The proof follows by similar arguments as in the proof of Theorem \ref{theo:phasespacefunc}.
\end{proof}
By slightly modifying the proofs of Theorem \ref{theo:classcausal} and Theorem \ref{theo:classtimeslice}
it is easy to show that the covariant functor $\PhaseSpace^0:U(1){-}\PrBu \to \PreSymp$
satisfies the classical causality property and the classical time-slice axiom. In addition,
we have have the following
\begin{theo}
The covariant functor $\PhaseSpace^0:U(1){-}\PrBu \to \PreSymp$ satisfies the locality property:\sk

Let $f:\Xi_1\to \Xi_2$ be any morphism in $U(1){-}\PrBu$, then $\PhaseSpace^0(f)$
is injective.
\end{theo}
\begin{proof}
Notice first that any element $[\varphi]\in  \PhaseSpace^0(\Xi_1)$ that satisfies
$[f_\ast(\varphi)] =0$ is necessarily contained in the radical $\mathcal{N}_1^0 \subseteq \PhaseSpace^0(\Xi_1)$.
Let us now assume that $[\varphi]\in \mathcal{N}_1^0$ is such that
$[f_\ast(\varphi)]=0$. By Lemma \ref{lem:presympzero} b) there exists a representative 
$\varphi\in \sectn{M_1}{\mathcal{C}(\Xi_1)^\dagger}$ of $[\varphi]$ that is of the form $\varphi = a\,\1_1$ with
 $a\in C^\infty_0(M_1)$. The push-forward along $f$ of this representative is then $f_\ast(a\,\1_1) = \underline{f}_\ast(a)\,\1_2$,
 where $\underline{f}_\ast(a)\in C^\infty_0(M_2)$ is the push-forward along $\underline{f}:M_1\to M_2$.
 Since by hypothesis $[f_\ast(\varphi)]=0$, the representative $\underline{f}_\ast(a)\,\1_2$ is equivalent to 
 an element in $\mathrm{Triv}_2$, i.e.~for some $\eta\in \Omega^2_{0,\dd}(M_2,\g^\ast)$ and $b\in C^\infty_0(M_2)$
 satisfying $\int_{M_2} \vol_2\,b=0$, we have $\underline{f}_\ast(a) \,\1_2 = b\,\1_2 + \underline{\mathcal{F}_2}^\ast(\eta)$.
 Comparing the linear parts of both sides of the equality we obtain $\delta_2\eta =0$, i.e.~$\eta\in \Omega^2_{0,\dd}(M_2,\g^\ast)$
 is both closed and coclosed. As a consequence, $\square_{2\,(2)}(\eta)=0$, which due to normal hyperbolicity implies that
 $\eta =0$. We find $\underline{f}_\ast(a)=b$ and in particular $0=\int_{M_2}\vol_2\,\underline{f}_\ast(a) = \int_{M_1}\vol_1\,a$.
 Thus, $[\varphi]=[a\,\1_1]=0$ since $a\,\1_1\in \mathrm{Triv}_1$.
\end{proof}
Let us denote by $\PreSymp^{\mathrm{inj}}$ the subcategory of $\PreSymp$ where all morphisms are injective.
We have shown above the existence of the covariant functor $\PhaseSpace^0:U(1){-}\PrBu \to \PreSymp^{\mathrm{inj}}$.
Since the $\CCR$-functor restricts to a covariant functor $\CCR:\PreSymp^{\mathrm{inj}}\to \astAlg^{\mathrm{inj}}$, where
we have used the obvious notation for the subcategory of $\astAlg$ with injective morphisms,
we obtain by composition a covariant functor $\QFT^0: U(1){-}\PrBu \to \astAlg^{\mathrm{inj}}$.
The classical causality property and the classical time-slice axiom extend via the $\CCR$-functor 
to the quantum case, see e.g.~\cite[Theorem 6.3]{Benini:2012vi}.
The main result of this section can be summarized as follows:
\begin{theo}
The covariant functor  $\QFT^0:= \CCR\circ \PhaseSpace^0: U(1){-}\PrBu \to \astAlg^{\mathrm{inj}}$
is a locally covariant quantum field theory, i.e.~$\QFT^0$ satisfies the quantum causality property, the quantum time-slice axiom
and the locality property.
\end{theo}


\section*{Acknowledgements}
We would like to thank Hanno Gottschalk, Thomas-Paul Hack, Ko Sanders and Jochen Zahn for useful discussions and comments.
Furthermore, we are grateful to the anonymous referees for their constructive comments and
suggestions to improve the manuscript.
The work of C.D.\ has been supported partly by the University of Pavia and partly
 by the Indam-GNFM project {``Effetti topologici e struttura della teoria di campo interagente''}. 
 The work of M.B.\ has been supported by a DAAD scholarship. 
M.B.\ is grateful to the II.\ Institute for Theoretical Physics of the University of Hamburg for the kind hospitality.


\end{document}